\documentclass[twocolumn,10pt]{article}          
\usepackage[numbers]{natbib}

\usepackage{makeidx}  
\usepackage{fullpage}
\usepackage{afterpage}
\usepackage{amsmath}
\usepackage{amstext}
\usepackage{amsbsy}
\usepackage{amssymb}
\usepackage[nodots]{numcompress}
\usepackage{amsthm}
\usepackage{graphicx}
\usepackage{subfigure}
\usepackage{url}
\usepackage{arydshln}

\usepackage{algorithm} \usepackage[noend]{algorithmic}
\algsetup{indent=2em}

\usepackage{prettyref}
\newtheorem{lemma}{\textbf{Lemma}}

\newcommand{\algo}[1]{\textsc{#1}} 
\newcommand{\var}[1]{\textsl{#1}} 

\newrefformat{alg}{Algorithm~\ref{#1}}
\newrefformat{eq}{Equation~\ref{#1}}
\newrefformat{lem}{Lemma~\ref{#1}}
\newrefformat{def}{Definition~\ref{#1}}
\newrefformat{cor}{Corollary~\ref{#1}}
\newrefformat{chp}{Chapter~\ref{#1}}
\newrefformat{sec}{Section~\ref{#1}}
\newrefformat{apx}{Appendix~\ref{#1}}
\newrefformat{tab}{Table~\ref{#1}} 
\newrefformat{fig}{Fig.~\ref{#1}}

\begin{document}

\title{1-D and 2-D Parallel Algorithms for All-Pairs Similarity
  Problem}

\author{Eray \"Ozkural, Cevdet Aykanat} 
 


\date{10 May 2013}

\maketitle 

\begin{abstract}
  All-pairs similarity problem asks to find all vector pairs in a set
  of vectors the similarities of which surpass a given similarity
  threshold, and it is a 
  computational kernel in data mining and information retrieval for
  several tasks.
  We investigate the parallelization of a recent fast sequential
  algorithm.
  We propose effective \mbox{1-D} and 2-D data distribution strategies that
  preserve the essential optimizations in the fast algorithm. 1-D
  parallel algorithms distribute either dimensions or
  vectors, whereas the 2-D parallel algorithm distributes data both ways. 
  Additional contributions to
  the 1-D vertical distribution include a local pruning strategy to 
  reduce the number of candidates, a recursive pruning algorithm, 
  and block processing to reduce imbalance.  
  The parallel algorithms were programmed in OCaml which affords much
  convenience.
  Our
  experiments indicate that the performance depends on the dataset,
  therefore a variety of parallelizations is useful.
\end{abstract}

\section{Introduction}

Given a set $V$ of $m$-dimensional $n$ vectors and a similarity
threshold $t$, the all-pairs similarity problem asks to find all
vector pairs with a similarity of $t$ and more.  Given the high
dimensionality of many real-world problems, such as those arising in
data mining and information-retrieval, this task has proven itself to
be quite costly in practice, as we are forced to use the brute-force
algorithms that have a quadratic running time complexity. Recently,
Bayardo et.~al \cite{Bayardo2007Scaling} have developed time and
memory optimizations to the brute force algorithm of calculating the
similarity of each pair in $V \times V$ and filtering them according
to whether the similarity exceeds $t$.  We may assume the vectors are
in $R^m$ and the similarity function is inner product without much
loss of generality. 

Two 1-D data distributions are considered: by
dimensions (vertical) and by vectors (horizontal). We introduce useful
parallelizations for both cases.  We have observed that the optimized
serial algorithms are suitable for parallelization in this fashion,
thus we have designed our algorithms based upon the fastest such
algorithm.  It turns out that our horizontal algorithms especially
attain a good amount of speedup, while the elaborate vertical
algorithms can attain a more limited speedup, partially due to
limitations in our implementation. Additional contributions to
the 1-D vertical distribution include a local pruning strategy to 
reduce the number of candidates, a recursive pruning algorithm, 
and block processing to reduce imbalance.  
We have combined the two data
distribution strategies to obtain an elegant 2-D parallel algorithm. We also
take a look at the performance of a previously proposed family of
optimized sequential algorithms and determine which of those
optimizations may be beneficial for a distributed memory parallel
algorithm design. 
We have implemented the parallel algorithms in the functional language
OCaml. 
A performance study compares the performance of the
proposed algorithms on small and large real-world datasets.

\subsection{Overview}

The rest of the paper is organized as
follows. \prettyref{sec:background} briefly gives the background of
the problem while in \prettyref{sec:relatedwork} we review related
work. Optimizations to the sequential algorithm are covered in
\prettyref{sec:optseq}.  \prettyref{sec:1dpar} introduces the 1-D
vertical and horizontal parallelizations, and likewise
\prettyref{sec:2dpar} presents the 2-D
parallelization. \prettyref{sec:perf} contains the performance study
and \prettyref{sec:conclusion} explains our conclusions and future work.

\subsection{Contributions}

Upon the first author's suggestion that two-dimensional algorithms
would be appropriate for frequent itemset mining and text
categorization problems, the second author contributed the idea that a
two-dimensional algorithm could work for the Euclidian all-pairs 
similarity problem (which the first author suggested as an important
problem as input to graph transduction algorithms), and offered a 
parallelization based on a mesh network. 
The first author later refined that approach by designing algorithms
that parallelize the recent fast
all-pairs similarity algorithms developed at Google. 
The first author then optimized the algorithms for non-blocking
networks, including the pruning and vector blocking approaches. The first
author proved the theorems, made the implementation in OCaml,
performed the experiments, and wrote the paper. The second author was
the PhD supervisor of the first author at the time of writing, and guided the
research by making valuable suggestions, and has endorsed this
paper. The second author also carefully reviewed and guided all
theoretical research on this problem and contributed the performance 
analysis framework in \prettyref{sec:perf}, which the first author extended. 

\section{Background}
\label{sec:background}

\subsection{Problem definition}

Let $V =\{ v_1, v_2, v_3, ..., v_n\}$ be the set of sparse input vectors in
$R^m$, following a similar terminology to \cite{Bayardo2007Scaling}.
Let $t$ be the similarity threshold.  Let a sparse vector $x$
be made up of $m$ components $x[i]$, where some $x[i]=0$; such a
sparse vector can be represented by a list of pairs $\left[(i, x[i])
\right]$ in which only non-zero components are stored.  Let $|x|$ be
the number of non-zero components in the vector, that is the length of
its list representation.  Let $||x||$ be the vector's magnitude. Let
also $size(V)=\sum_{v \in V}|v|$ be the number of non-zero values in
$V$.  Each vector $v_i$ is made up of components per dimension $d$,
where the vector's $d$th component is denoted as $v_i[d]$. The
similarity function is defined as the summation of input values from
similarity among individual components:
$sim(x,y)=\sum_{i}sim(x[i],y[i])$. Another accumulation function
instead of summation may be used (for instance any other binary
operation which has the same algebraic properties), however summation
is enough for many purposes.  The problem is to find the set of all
matches $M=\{ (v_i,v_j) \ |\ {v_i \in V} \land {v_i \in V} \land {i \neq j}
\land {sim(v_i,v_j) \geq t} \}$.

Without much loss of generality, we assume that input vectors are
normalized (for all $x \in V, ||x||=1$), and for vectors $x$ and $y$,
$sim(x,y)$ function is the dot-product function $dot(x,y) =
\sum_{i}x[i] . y[i]$, that is $sim(x[i],y[i])=x[i].y[i]$. The
algorithms can be easily generalized to other similarity functions
which are composed from similarities $sim(x[i], y[i])$ across
individual dimensions.  

The input dataset $V$ may also be interpreted as a data matrix $D$
where row $i$ is vector $v_i$. In this case, we may represent
similarities by the similarity matrix $S = D.D^T$ where
$S_{ij}=dot(v_i,v_j)$ obviously, and we find the set of matches $M =
\{(i,j) \ |\ S_{ij} \geq t \}$. More naturally, we may interpret the
output as a match matrix $M$ that is defined as:
\begin{equation}
  M'_{ij} = 
  \begin{cases}
    0 & \text{if} \ S_{ij} < t, \\
    S_{ij} & \text{if} \ S_{ij} \geq t
  \end{cases}
\end{equation}

The output set of matches $M$ may be considered to define an
undirected similarity graph $G_S(V, t) = (V, M)$. In this case an edge
$ u \leftrightarrow v$ denotes a similarity relation between vectors
$u$ and $v$; the edge weight $w(u,v) = u.v$. 

\subsection{Applications}

An all pairs similarity algorithm may be viewed as a computational
kernel for several tasks in data mining and information retrieval
domains.  In data mining and machine learning, the similarity graph
 may be supplied as input to efficient graph
transduction \cite{Joachims03transductive,Wang08graph}, graph
clustering algorithms \cite{brandes2003} and near-duplicate detection
(by using a high threshold to filter edges). Obviously, once a
similarity graph is computed, classical k-means
\cite{macqueen67,lloyd57} or k-nn algorithms \cite{fix51,cover67},
which are widely used in data mining due to their effectiveness in low
number of dimensions, may be adapted to use the graph instead of
making the geometric calculations directly over input vectors.  As
frequent itemset mining may be viewed as the costly phase of
association rule mining class of algorithms; likewise, the graph
similarity problem may be viewed as the costly phase of several
classification, transduction, and clustering algorithms. 

Calculating
the similarity graph may be alternately viewed as capturing the
essential geometry of (the similarities in) the dataset, on which any
number of computational geometry algorithms may be run. This is
basically what a classification or clustering algorithm does given
similarities in the data: the algorithm tries to find geometric
distinctions, either determining a class boundary for classification,
or identification of clusters by grouping similar points according to
the similarity geometry.  Note also that with an adequate similarity
threshold, we can obtain a connected graph and therefore approximate
\emph{all} similarities in the dataset.  

Constructing the similarity
graph also has the unique advantage in that it can be re-used later
for additional data mining tasks. For instance, one application can
make a hierarchical clustering of the data, and another one can use it
for transduction. Basically, we think that any data mining task that
has a \emph{geometric} interpretation can use the similarity graph as
input successfully. Therefore, we anticipate that the parallel
similarity graph construction will be a staple of future parallel data
mining systems.

\section{Related Work}
\label{sec:relatedwork}

\subsection{k-nearest neighbors problem}

The problem of constructing a similarity graph can be contrasted with
k-nearest neighbors problem, which is a slightly harder problem but
can be solved approximately using a distance threshold. Our use of the
dot-product between two vectors should not be misleading either, as
that corresponds to range search in a corresponding metric space, to
emphasize the close relation between these problems.  At any rate,
some of the same approaches can be adapted to similarity graph
construction, therefore we should take them into account. Especially,
note that most of the difficulties with nearest neighbor search carry
over to our problem.

Due to the curse of dimensionality \cite{marimont79}, the brute-force
algorithm of nearest neighbor search is quite difficult to improve
upon \cite{weber1998}. In practice, there are no advanced geometric
data structures that will give us algorithmic shortcuts
\cite{Chavez2001,asano95}. In the general setting of metric spaces,
the nearest neighbor problem is non-trivial and data structures are
not very effective for high dimensionality \cite{clarkson06}.  This
implies that we cannot rely on space partitioning or metric data
structures that work well in low number of dimensions, although of
course, non-trivial extensions of those methods may prove to be
effective such as combining dimensionality reduction with geometric
data structures.

\subsection{Related sequential algorithms}

\subsubsection{Sequential knn algorithms}

Some popular approaches to solving the nearest neighbor problem may be
summarized as geometric data structures such as
R-Tree\cite{Guttman84r-trees}; VP-Tree \cite{Yianilos93}, GNAT
\cite{Brin:1995} and M-Tree \cite{Ciaccia:1997} for general metric
spaces, pivot-based algorithms \cite{Farago:1993,Bustos2010}, random
projections for $\epsilon$-approximate solutions to the knn problem
\cite{Kleinberg1997}, combining random projections and rank
aggregation for approximation \cite{Fagin2003}, locally sensitive
hashing \cite{Gionis1999,Ailon2006,Andoni2008}, and other data
structures and algorithms for approximations
\cite{Kushilevitz1998,Arya94anoptimal}.  An algorithm related to our
area of interest detects duplicates by using an inverted index
\cite{Ilyinsky2002}. Space-filling curves have also been applied to
the knn problem
\cite{rafajlowicz-knn-sfc,Liao00highdimensional,Mainar-Ruiz:2006}.

Space-partitioning approaches usually do not work well for very
high-dimensional data due to the curse of dimensionality, a thorough
treatment of which is available in \cite{weber1998}. Weber et. al
quantify in that article lower bounds on the average performance of
nerarest neighbor search for space and data partitioning assuming
uniformly distributed points, which show that for space partitioning
like k-d trees, the expected NN-distance grows with increasing
dimensionality, rendering such methods ineffective for
high-dimensional data (full scan needed when $d>60$), and for
data-partitioning the number of blocks that have to be inspected
increase rapidly with increasing number of dimensions, for both
rectangular (full scan is faster when $d>26$) and spherical bounding
regions (full scan when $d>45$), and they also generalize their
results to any clustering scheme that uses convex clusters, not just
these. Their conclusion is that in high-dimensional data, the
partitioning methods all degenerate to sequential search, in uniformly
distributed data.  We emphasize that their results imply that trivial
geometric partitions of the data using hyperplanes or hyperspheres are
mostly ineffective in very high-dimensional data, although they can in
some cases work well for datasets with limited dimensionality or
different distribution.  Weber et. al for this reason propose the
VA-file, which approximates vectors using bitstrings \cite{weber1998}
and improves upon sequential scan. 

In general, it seems that for
solving proximity problems exactly in very high-dimensional datasets,
techniques that prune candidates work well. Kulkarni and Orlandic, on
the contrary, successfully use a data clustering method to optimize
knn search in databases, which the authors show to be better than
sequential scan and VA-file up to 100 dimensions on random datasets
and 56 dimensions on real-world datasets \cite{kulkarni06}, although
it is impossible to know the true efficiency of these algorithms
proposed by database researchers unless they are compared to fast
in-memory algorithms since disk access time dominates the running time
of algorithms that work on secondary storage. Also, such approaches do
not usually scale up to very high number of dimensions.

Note that there are asymptotically optimal nearest neighbor algorithms
in the literature. Vaidya introduces an asymptotically optimal
algorithm for the all nearest neighbors problem which has $O(n \log
n)$ time complexity \cite{vaidya1989}. The same algorithm solves
$k$-nearest neighbors problem in $O(n \log n + k n \log k)$ time,
while Callahan and Korasaraju propose an optimal $k$ nearest neighbors
algorithm which runs in $O(n \log n + kn)$ time
\cite{Callahan92adecomposition}. It is not immediately obvious why
there are no experiments measuring the real-world performance of these
optimal algorithms, however, it is conceivable that they may not have
been practical for high-dimensional datasets, or it may have been 
considered that they require large constant factors.

We refer the reader to Chavez's survey of search methods in metric
spaces \cite{Chavez:2001} for more information on the myriad
algorithms. Chavez identifies three kinds of search algorithms for
metric spaces: pivot-based algorithms, range coarsening algorithms,
and compact partitioning algorithms, and he emphasizes that the search
time of exact algorithms grow with intrinsic dimensionality of the
metric space, which also increases the search radius, and thus makes
it harder to compete with brute-force algorithms. As we have seen,
similar problems also plague search algorithms in Euclidian spaces.
For these reasons, researchers in recent years have turned to
practical optimizations over brute-force algorithms, which we shall
now examine briefly with a good example.

\subsubsection{Practical sequential similarity search}
\label{sec:practicalseq}

In Bayardo et.~al \cite{Bayardo2007Scaling}, the authors propose three
main algorithms which embody a number of heuristic improvements over
the quadratic brute force all-pairs similarity algorithm. 
These algorithms are summarized below. In
the algorithms, each vector $x$ has components with weights $x[i]$,
there are $m$ dimenions (or features) numbered from $1$ to $m$,
$maxweight_i(V)$ is the maximum weight in dimension $i$ of the entire
dataset $V$, and $maxweight(x)$ is the maximum weight in a vector $x$,
following the notation in their paper.

\begin{description}

\item[all-pairs-0] This is equivalent to the brute force algorithm,
  with the additional on-the-fly construction of an inverted index as
  each vector is matched and indexed in turn. The calculation of the
  dot-product scores are achieved by consulting the inverted
  index. Thus each vector is compared to all the previous vectors that
  have been indexed, and then the vector itself is added to the
  index. This algorithm is thus slower than the brute force
  algorithm. In the matching of a new vector $x$, the algorithm uses a
  hash table $A$ to store the weights of candidates to match against
  $x$, since the vectors are sparse. 
  The pseudocode for \var{all-pairs-0} 
  is given in \prettyref{alg:all-pairs-0} and
  \prettyref{alg:find-matches-0}.

  \begin{algorithm}[top]
    \caption{$\algo{All-Pairs-0}(V,t)$}
    \label{alg:all-pairs-0}
    \begin{algorithmic}
      \STATE $M \gets \emptyset$ 
      \STATE $I \gets \algo{Make-Sparse-Matrix}(m, n)$ 
      \FORALL{$v_i \in V$} 
        \STATE $M \gets M \cup \algo{Find-Matches-0}(v_i, I, t)$ 
           \FORALL{$v_i[j]$ where $v_i[j]>0$} 
               \STATE $I_{ji} \gets v_i[j]$
           \ENDFOR
       \ENDFOR
      \RETURN $M$
    \end{algorithmic}
  \end{algorithm}

  \begin{algorithm}[top]
    \caption{$\algo{Find-Matches-0}(x,I,t)$}
    \label{alg:find-matches-0}
    \begin{algorithmic}
      \STATE $A \gets \algo{Make-HashTable()}$ \FORALL{$(i,x[i]) \in
        x$ where $x[i] \neq 0$} \FORALL{$(y, y[i]) \in I_i$}
      \STATE{$A[y] \gets A[y] + x[i]. y[i] $}
      \ENDFOR
      \ENDFOR
      \RETURN $\{ (y, A[y]) \ | \ A[y] \geq t \}$
    \end{algorithmic}
  \end{algorithm}

\item[all-pairs-1] This algorithm orders the dimensions in the order
  of decreasing number of non-zeroes. It corresponds to an important
  optimization that we call ``partial indexing'' which works as
  follows. In preprocessing, we calculate $maxweight_i(V)$ for each
  dimension. This allows us to calculate an upper bound for the
  dot-product of a vector $x$ with any vector in $V$: $ \forall y \in
  V$ $x.y \leq \sum_{i}x[i].maxweight_i(V)$. Using this upper bound it
  is possible to avoid indexing the most dense dimensions by
  calculating a partial upper bound $b$ while processing the
  components of new vector $x$ for indexing. Remember that we are
  processing the components in a certain order (decreasing number of
  nonzeroes of dimensions in $V$). The components are added to the
  inverted index only when the partial upper bound $b$ exceeds $t$,
  the initial components that have small $b$ are not indexed at all,
  they are kept as a partial vector $x'$. Indexing as such ensures
  that all admissible candidate pairs are generated. The dot-product
  is fixed by adding the dot-products of the partial $x'$'s later on.
\item[all-pairs-2] This algorithm affords three optimizations over
  \var{all-pairs-1}.

  \emph{Minsize optimization:} This optimization aims to prune
  candidate vectors with few components.  We know that for a vector
  $x$, for all matches y, $x.y \geq t$.  If the input vectors are
  normalized, then each component can be at most $1$: $x.y <
  maxweight(x) . |y|$. Two inequalities entail that $|y|\geq
  t/maxweight(x)$. Let the quantity on the right be called
  \emph{minsize}'.  Minsize optimization requires the vectors to be
  ordered in order of increasing $maxweight(x)$, thus decreasing
  \emph{minsize}.  If ordered such and the input vectors are normalized,
  during matching a new vector $x$, the minimum size of a candidate
  vector $y$ that x can be matched against is $t/maxweight(x)$. If the
  candidates in the inverted index that are smaller than \emph{minsize} are
  pruned when matching a new vector, this will hold true for all the
  subsequent vectors since \emph{minsize} for subsequent vectors cannot be
  greater. The minsize optimization does not prune a lot of
  candidates, but it may be effective since there may be a lot of very
  small vectors.  It is suggested that \var{all-pairs-2} prunes only vectors
  in the beginning of the inverted list, which is easy to implement
  using dynamically sized arrays.

  \emph{Remscore optimization:} This optimization calculates a maximum
  remaining score (remscore) while processing the components of a
  vector $x$ during matching, using $maxweight_i(V)$ function. When
  remscore drops below $t$ the algorithm switches to a strategy that
  avoids adding any new candidates to the candidate map, while
  continuing to update the candidates already in the map. This avoids
  calculation of scores for candidates that cannot match. Remscore is
  initialized as $\sum_ix[i].maxweight_i(V)$ and as each component $i$
  is processed its contribution to the upper bound
  $x[i].maxweight_i(V)$ is subtracted from the upper bound. And while
  calculating the scores in the candidate map, the aforementioned
  conditional is executed. While this seems to be an excellent
  optimization, in the real-world data we have seen it has only
  inflated the running time, because not the calculation of remscore but
  the conditional reasoning is too expensive within the main loop of
  matching algorithm.

  \emph{Upperbound optimization:} While fixing the scores in the
  candidate map with dot-products of partial vectors (parts of vectors
  that are not indexed), we can avoid the dot-product if the following
  upper bound is not enough to make the score exceed $t$: $min(|y'|,
  |x|). maxweight(x) . maxweight(y')$ which is to say that each scalar
  product in an inner product cannot be more than the product of the
  maximum values in either vector, and only non-zero components
  contribute to the inner product. While this too seems to be a nice
  optimization, it suffers from using conditionals in an otherwise
  efficient code as the partial vectors tend to be short.

\end{description}

\subsubsection{Analysis of all-pairs-0}

\var{All-pairs-0} maintains an inverted index $I$, which
stores an inverted list for each of $m$ dimensions in the dataset,
such that after all the matches are found, for a vector $v_i$ and for all
$v_i[j]$, the inverted index $I$ stores $v_i[j]$, that is $I_{ji} = v_i[j]$.

If the inverted index $I$ is interpreted as a matrix, the rows $I_j$
of the inverted index are the dimensions in the dataset, and $I$ is
merely the transposition of the input matrix $D$, $I=D^T$.  Algorithm
\var{all-pairs-0} performs $\sum_{d=1}^{m}\binom{|I_d}{2}$ floating-point
multiplications, dominating the running time complexity, therefore
each dimension $d$ contributes $\binom{|I_d|}{2}=O(|I_d|^2)$
multiplications.

Since in practice there are usually a few dense dimensions, the running
time complexity is expected to be quadratic in $n$ for real-world
datasets.

\subsection{Related parallel algorithms}

There are only a few relevant studies on efficient parallelization of
the all pairs similarity problem in the literature that we have been
able to detect.

Lin \cite{lin2009} parallelizes the all-pairs similarity problem comparing
parallelizations of both the brute force algorithm that uses no
intermediate data
structures and two algorithms that use an inverted index of the data,
one horizontal and one vertical parallelization (called Posting
Queries and Postings Cartesian Queries algorithms), implemented with
the map/reduce framework Hadoop.  The
algorithm is cast in an information retrieval context where
documents are vectors and terms are dimensions. The experiments are
quite comprehensive and utilize realistic life sciences datasets. The
study in question also compares the performance of three approximate 
solutions: limiting the number of accumulators, considering only top
$n$ terms in a document, and omitting terms above a document frequency
threshold; their results show that significant performance gains can
be obtained from approximate solutions at acceptable loss of
precision. Therefore, Lin suggests that parallelizing the exact
algorithms easily carry over to more efficient inexact 
algorithms. However, there is a slight drawback of this careful study,
as the use of Java language may have caused significant performance
loss in the sequential algorithms, making the job of parallelization
easier, as for 90 thousand documents, their sequential algorithm takes
on the order of hundreds of minutes on a cluster system. Lin does
mention that the code is not optimized and run on a shared,
virtualized environment. In our experience, shared environments are
not suitable for working on memory and communication intensive
problems such as those in information retrieval and data mining. Thus,
we are looking forward to the repetition of the said experiments on a
dedicated parallel computer with a more appropriate high-performance
implementation. This study is also important in that the author
correctly observes the influence of the Zipf-like distribution of
terms on parallel performance.

Recently Awekar et.~al \cite{awekar2010} introduced a task parallelization of the all
pairs similarity problem, sharing a read-only inverted index of the
entire dataset.  The authors use a fast sequential
algorithm which is very similar to our \var{all-pairs-0-array}, which we
also found to be the best sequential algorithm, and thus make adequate
speedup measurements. The authors test three load balancing
strategies, namely block partitioning, round-robin partitioning, and
dynamic partitioning on high-dimensional sparse data-sets with a power
law distribution of vector sizes. Their experiments are executed on up
to 8 processors for large real-world datasets, on both a shared-memory
architecture and a multi-processor system. The speedups on the
multi-processor system turn out to be superior to the shared memory
system as cache-thrashing and memory-bandwidth limitation prevents
near-ideal performance for larger number of processors on
shared-memory systems.  In this study~\cite{awekar2010},  however, there is a major
shortcoming as the index construction and replication costs were not
taken into account in the experiments, which raises doubts as to how
much time is needed for broadcasting such large datasets (e.g., Orkut
dataset has 223 million non-zeroes), as the replication of the entire
inverted index would be a bottleneck for high number of
processors. Therefore, the replicated index algorithm should be taken
with a grain of salt, as well as any parallel algorithm that
replicates the entire dataset, since the size of the inverted index is
the same as the size of the dataset.  At any rate, near-ideal speedup
on up to $8$ processors is not surprising as our vector-wise 
parallelization shows similar performance, as will be seen.

Following are parallelizations of related problems.  Plaku and
Kavraki propose a distributed, message-passing algorithms for
constructing knn graphs of large point sets with arbitrary distance
metric \cite{plaku-distknn}. They can use any local knn data structure
for faster queries (such as a metric tree), which must be built once
the points are distributed to processors.  In addition to this, they
can exploit the triangle inequality of metric function and this
information can be used to construct local queries using the metric
data structure as well as pruning distributed queries, by representing
the bounding hyperspheres of points on other processors.  The
dimensionality of their datasets increases to non-trivial numbers (up
to 1001), and their speed-up results on 100 processors are quite
encouraging.  We think that their method might be applied to our work
as well in the future, to optimize our horizontal parallel algorithms,
however the effectiveness of their approach on very high-dimensional
datasets as we are using remains to be seen, as no sort of space
partitioning usually works well for very high-dimensional datasets due
to the curse of dimensionality.
However, it is conceivable that the methods of Plaku and Kavraki could
be used in hybrid approaches to deal with much higher dimensionality.
A shortcoming of this paper is that it does not discuss the
partitioning of the point set, any partition is assumed.

Alsabti et al. \cite{alsabti98} parallelize all pairs similarity search with a k-d tree
variant using two space-partitioning methods based on quantiles and
workload; they find that their method works well for 
12-d randomly generated points on up to 16 processors. Their workload
based partitioning scales better than quantile based partitioning, and
is comparable for uniform and gaussian distributions.  Apar\'{i}cio
et.~al \cite{aparicio2007} use a three-level parallelization of knn
problem at the Grid, MPI and shared memory levels and integrate all three to optimize
performance.  An interesting paper proposes a
parallel clustering algorithm which partitions a similarity graph,
constructs minimum spanning trees for each subgraph and then merges
the minimum spanning trees, which is then used to identify clusters
\cite{olman2009}; this algorithm can be applied to the output of our
algorithms. Schneider \cite{schneider1989} evaluates four parallel join algorithms for
distributed memory parallel computers from a database perspective. 
Vernica et~al. \cite{vernica2010}  propose a three-stage map/reduce
based approach to calculate set similarity joins and report results
using Hadoop; they do consider the self-join case.

Callahan and Kosaroj \cite{Callahan92adecomposition,callahan1995} 
examine the well-separated pair decomposition of
a point set in Euclidian space, which decomposes the set
of all pairs in a point set into pairs of sets with the constraint of
well-separation (defined in a certain geometric sense), wherein each
pair is uniquely represented by a pair of point sets in the
decomposition.  Using their
decomposition, they also obtain an asymptotically optimal parallel knn
algorithm which has $O(\log^2n)$ total parallel time on $O(n)$
processors with the CREW PRAM model.  The real-world
applicability of this wonderfully efficient algorithm remains to be
seen, however. In our initial inspection, we have seen their splitting
logic may be somewhat problematic in text data sets where each
co-ordinate corresponds to the frequency of a term.  It seems that one
way such space decomposition based algorithms may escape the curse of
dimensionality is that the decomposition is far from random, and that
the distribution is not uniform in real-world datasets, although one
may still expect that the approach might break down in very
high-dimensional datasets as their approach is conceptually similar to well
known k-d tree construction algorithms that fail in high-dimensional datasets.

\section{Optimizations to the sequential algorithm}
\label{sec:optseq}

In this section, we examine the optimizations in the sequential
algorithms of \prettyref{sec:practicalseq} detail, as they influence
our parallel algorithm design. 
We have made several other versions of these algorithms
to understand the impact of individual optimizations. This has aided
us in understanding the advantages and disadvantages of said
optimizations and design parallel algorithms. The slowness of
\var{all-pairs-2} compared to \var{all-pairs-1} on our datasets urged us to
understand the impacts of optimizations better.  
\begin{description}
\item[all-pairs-0-array] Although the input vectors are sparse, some
  dimensions are dense in the real-world data that we are using. Thus,
  the hash table $A$ is in fact dense. Using an array instead of a
  hash table improves running time.
\item[all-pairs-0-array2] Tries to optimize \var{all-pairs-0-array} further
  by maintaining a list of candidate indices that are used during
  matching, which are zeroed before finding the matches of the next
  vector.
\item[all-pairs-0-remscore] remscore optimization added to \var{all-pairs-0}
\item[all-pairs-0-minsize] minsize optimization added to \var{all-pairs-0}
\item[all-pairs-1-remscore] remscore optimization added to \var{all-pairs-1}
\item[all-pairs-1-upperbound] remscore optimization added to \var{all-pairs-1}
\item[all-pairs-1-minsize] minsize optimization added to \var{all-pairs-1}
\item[all-pairs-1-remscore-minsize] minsize and remscore optimizations added
  to \var{all-pairs-1}
\item[all-pairs-bruteforce] Brute force algorithm that uses no
  intermediate data structures
\end{description}

The performance comparison of the various implementations on two
datasets is given in \prettyref{sec:seqperf}, in which we see that
\var{all-pairs-0-array} is the fastest implementation, therefore we focus on
parallelizing that algorithm and for the remainder of the paper ignore
other algorithms. Note that on another software platform, perhaps one
of the other variants could be as efficient as \var{all-pairs-0-array}, however,
we think that the wide performance gap would be non-trivial to close.

\section{1-D Parallel Algorithms}
\label{sec:1dpar}

In the following parallel algorithms, let $p$ be the number of
processors and $pid$ be the processor ID of the current processor.  We
will explain our dimension-wise and vector-wise parallelizations,
respectively. We call the dimension-wise parallelization vertical, and
vector-wise parallelization horizontal, for brevity and in analogy
with the matrix representation $D$ of input where each vector is a row.

\subsection{Vertical algorithm: partitioning dimensions}

In vertical parallelizations, each processor holds a number of
dimensions (features), considered to be weighed by the square of
number of nonzeros as for finding the matches of each vector, the
entire inverted list of a dimension has to be scanned, and its
contribution to the candidate matches calculated. 

Each dimension $d$ contributes
\begin{displaymath}
w[d] = |I_d|.(|I_d|+1)/2
\end{displaymath}

multiplications, and thus the entire work may be assumed to take $w =
\sum_i w[i]$.

Since the dimensions are
split across processors, each inverted list is stored wholesome.
To iterate, each dimension has a home processor and each inverted list
corresponding to that dimension also has the same home processor. 
Therefore, each processor is responsible
for calculating the matches in a subspace composed of the dimensions
assigned to it.

Our vertical parallel algorithms essentially parallelize the
inner loop (find-matches phase) of the \var{all-pairs-0-array} algorithm,
while maintaining the sequential order of processing vectors.
Therefore, much attention is devoted to efficient processing of
separate subspaces and merging the candidates, which is the main
parallel overhead of this parallelization.

\subsubsection{Initial distribution}

The simplest distribution is cyclic distribution of dimensions, which
is a random distribution of dimension, however it has turned out to
result in too much load imbalance. Therefore, we use the following
simple partitioning algorithm. The dimensions are sorted in order of
decreasing non-zeroes and the dimensions are binned to $p$ bins so as
to balance the load. To achieve this, we use a first-fit algorithm
that places the next dimension in the least loaded processor.  We
distribute the dimensions before starting and timing the parallel
algorithm.

\subsubsection{Inner-loop parallelization of all-pairs-0}

\prettyref{alg:par-all-pairs-0-vert} depicts the pseudocode for the
basic vertical parallelization of \var{all-pairs-0} kind of algorithms. The
$comm$ variable is the MPI communicator used in the collective
communication operations, it is given as a variable to make the
algorithm re-usable in the 2-D algorithm.  In
\algo{par-all-pairs-0-vert}, first, we calculate the global number of
dimensions by taking the maximum among all processors. Then, we call
the parallel find-matches algorithm for each input vector $x$, which
calculates separate candidate maps on all processors and then
accumulates the candidate scores in parallel before filtering the
candidates.  Each processor thus computes partial candidate scores
independently and synchronously.  Then, scores are accumulated via
collective communication, which results in each processor having a
disjoint set of scores to filter, and the filtering is performed in
parallel.

\begin{algorithm}[top]
  \caption{$\algo{Par-All-Pairs-0-Vert}(V,t, comm)$}
  \label{alg:par-all-pairs-0-vert}
  \begin{algorithmic}
    \STATE $M \gets \emptyset$ \STATE $I \gets
    \algo{Make-Sparse-Matrix}(m, n)$ \FORALL{$x=v_i \in V$} \STATE $M
    \gets M \cup \algo{Par-Find-Matches-0-Vert}(x, I, t)$
    \FORALL{$x[j] \in V$ where $x[j]>0$} \STATE $I_{ji} \gets x[j]$
    \ENDFOR
    \ENDFOR
    \RETURN $M$
  \end{algorithmic}
\end{algorithm}

\subsubsection{Local pruning optimization}
\label{sec:localpruning}

We propose a local pruning optimization for the matching phase.
The parallelization of the inner loop is shown in
\prettyref{alg:par-find-matches-0-vert}. We employ local pruning to
decrease the number of candidates accumulated by collective
communication operations. Let us define $t_{local}$, the local
similarity threshold.
\begin{equation}
  \label{eq:1}
  t_{local} = t/p
\end{equation}
\begin{lemma}
  \label{lem:localpruning}
  Observe that, for any distribution of dimensions, if a candidate is
  matched, that is $sim(x,y) \geq t$, then the local similarity of at
  least one processor should be at least $t_{local}$.
\end{lemma}
\begin{proof}
  Assume that for all $p$ processors, the local similarities $sim(x,y)
  < t_{local}$. Then, obviously, $sim(x,y) < p . t / p$, that is
  $sim(x,y) < t$ which is a contradiction. Therefore, on at least one
  processor, the local similarity is greater than $t_{local}$.
\end{proof}

Making use of this lemma, on each processor we compute the array $A$
of local scores of $x$, and a set of local candidates $C$ which are
the candidates that meet local threshold $t_{local}$
effortlessly. These local scores and candidates are then merged using
a parallel score accumulation algorithm called
\algo{Accumulate-Scores-Vert}.

Note that we use arrays for candidate map instead of a hash table
because it is more efficient in practice.

\begin{algorithm}[top]
  \caption{$\algo{Par-Find-Matches-0-Vert}(x,I, t, comm)$}
  \label{alg:par-find-matches-0-vert}
  \begin{algorithmic}
    \STATE $t_{local} \gets t/p$ \STATE $A \gets \algo{Make-Array}(n)$
    \text{such that} $A[i] = 0$ \FORALL{$(x,x[i]) \in x$} \FORALL{$(y,
      y[j]) \in I[i]$} \STATE{$A[y] \gets A[y] + x[i]. y[i] $}
    \ENDFOR
    \STATE $C \gets \{ (x, y) | A[y] \geq t_{local} \}$ \RETURN
    $\algo{Accumulate-Scores-Vert}(A, C, comm)$
    \ENDFOR
  \end{algorithmic}
\end{algorithm}

\subsubsection{Score accumulation with local pruning}

The scores are accumulated in two communication steps. In the first
step, we perform an all-reduce operation using the binary
operation of set union. At the end of this step, every processor
obtains a $C_g$ of \emph{global} candidates. After this step, since
every processor already has the local scores $A$, which contain all
the local candidates in $C_g$, we take the local scores in $A$ which
are in $C_g$ and put them into a sparse vector $A'$.  On each
processor, for each candidate vector $y$ with weight $w$, we have
\begin{displaymath}
A'[y]=A[y]=w>t_{local} .
\end{displaymath}
 Succeeding that, we compute $A_g$ which is
the summation of sparse vectors on each processor, with the result
partitioned over all processors, so each processor stores a range of
indices of $A_g$. That is, we use a parallel sparse vector addition
algorithm with input and output partitioning.  Thereafter, the $A_g$
can be filtered in parallel to find scores that are at least $t$.

\subsubsection{Recursive local pruning}

In practice, local pruning works quite effectively on two processors,
but due to the nature of observed power-law like distribution of term
frequencies, every binary subdivision almost doubles the number of
candidates. If no local pruning is applied, we have observed that
about $n/2$ candidates are required on the average. With local
pruning, we observe a significant reduction of that number on two
processors (about 10-fold) making the vertical partitioning
competetive with horizontal partitioning.

By observing that local pruning can be applied recursively, we can
decrease the communication volume of the score accumulation
further. Let the dataset matrix $D$ be vertically partitioned $\Pi(D)
= \{D_1,D_2\}$ into roughly equal number of dimensions. Local pruning
result  (\prettyref{lem:localpruning}) entails that the set of candidates is
the union of all similar pairs in both sub-datasets with $t/2$. That
is to say, we obtain a set of candidates by taking the union of local
matches with $t/2$ threshold: $C(D, t) \supset M(D, t) = M(D_1,t/2)
\cup M(D_2, t/2)$. This process can be applied recursively. For
instance, another level of application would yield: $C(D_1, t/2) =
M(D_{11},t/4) \cup M(D_{12}, t/4)$ and $C(D_2, t/2) = M(D_{21},t/4)
\cup M(D_{22}, t/4)$ where $\{ D_{11}, D_{12} \}$ is a vertical
partition of $D_1$ and $\{ D_{21}, D_{22} \}$ is a vertical partition
of $D_2$.

This recursive sub-division suggests an algorithm. We first
recursively partition the dimensions in $k$ levels of recursion. At the bottom level
$k$ of recursion, we can find the matches for $M(D_p, t/2^k)$ where
the dataset label $p$ has $k$ numerals, and communicate these
pair-wise to calculate their union as the $2^{k-1}$ candidate sets for
the higher level. Now, we must compute the matches in the higher level
to calculate the yet higher level candidates and so forth, until we
have the candidates $C(D,t)$ for $D$. The intention here is that,
instead of broadcasting all the bottom level candidates, we are
communicating less. After computing candidates this way, another pass
could be used for score accumulation, but interleaved execution of
candidate generation and score addition steps would be faster.

In our example, consider that we have the candidates $C(D_1, t/2)$ and
$C(D_2, t/2)$ after the first two candidate union operations at the
bottom level $2$. We need a fast method to filter these candidate
sets, and the processors corresponding to $D_1$ must co-operate to
calculate $M(D_1, t/2)$ and likewise for $D_2$. If the candidates are
partitioned over processors in this step, the score accumulation can
be performed fast in parallel. Therefore, we split the candidate set
according to vector indices and communicate scores so that each
processor making up $D_1$ has a portion of the global scores, which
then it can filter to find its portion of $M(D_1, t/2)$. Note that 
this is also a partial score which may be useful to us later on, so we
store it.  Then, yet, in the top level when calculating the matches
$M(D_1, t)$, score accumulation can proceed among processors with
matching vector ranges of scores.

An important consideration in this algorithm is to be able to complete
partial scores. For instance, a vector $x$ may not be a candidate in
$D_1$, $D_{11}$ and $D_{12}$, but it may be a candidate in
$D_2$. Since it wasn't a candidate in any of the candidates in the
recursion sub-tree corresponding to $D_1$, the non-zero scores of $x$
in $D_1$ would have to be added. This requires knowing which
processors contributed to a score, and if there are missing we have to
send requests to those processors and get the missing information. The
processors in a partial sum may be represented with a bit-vector.

\subsubsection{Functional recursive local pruning algorithm}

In \prettyref{alg:merge-scores-rec}, we give a straightforward
functional algorithm for realizing recursive local pruning.  We assume
that we have $p=2^k$ processors, and we apply vertical
bi-partitioning recursively. The following algorithm assumes that
local scores have ben computed on each processors in array $A$.  $x$
is the vector for which matches are sought, and $t$ is the similarity
threshold; $i..j$ denotes an inclusive integer range.

\begin{algorithm}
  \caption{$\algo{Merge-Scores-Rec}(x, A, t, comm)$}
  \label{alg:merge-scores-rec}
  \begin{algorithmic}[1]
    \STATE $ pid \gets \algo{MPI-Rank}(comm)$ \STATE $ p \gets
    \algo{MPI-Size}(comm)$ \IF{p=1} \STATE $M \gets \{ (y, A[y]) \ | \
    {y \in (0...|A|-1)} \land {A[y]>t} \}$ \ELSE \STATE $color \gets $
    if $pid \in (0, p/2-1)$ then $0$ else $1$ \STATE $comm' \gets
    \algo{MPI-Comm-Split}(comm, color, pid)$ \STATE $M' \gets
    \algo{Merge-Scores-Rec}(x, a, t/2, comm')$ \STATE $C \gets
    \algo{Reduce-All}(comm, M', \algo{Union})$ \STATE $AL \gets \{ (y,
    A[y]) \ |\ {y \in C} \land {A[y]>0}\}$ \STATE $AG \gets
    \algo{Accumulate-Scores}(comm, AL)$ \STATE $M \gets \{ (y, w) \ |\
    {(y,w) \in AG} \land {w>t}\}$ \RETURN $M$
    \ENDIF
  \end{algorithmic}
\end{algorithm}

\subsubsection{Flat accumulation algorithm}
\label{sec:flataccum}

Alternatively, we can implement \algo{Accumulate-Scores-Vert} using
the MPI Allgather function for constructing the set union of local $C$
sets, and we can compute $A_g$, where each candidate vector is stored
on a processor. We can compute $A_g$ in distributed fashion by using
$p$ MPI Gather calls, and then locally adding the partial scores
across dimensions.  This is a practical implementation we are using in
our experiments on compute clusters, however a more scalable
implementation may be also developed in the future.

\subsubsection {Hypercube accumulation algorithm}
For accumulation, we can utilize an algorithm inspired by the parallel
quicksort algorithm on hypercube topology. The input to the parallel
accumulation algorithm is an association list of vector id, score
pairs for the current vector $x$.  Each association list is sorted in
the order of vector id's. In the partitioning step of the quicksort-like
accumulation algorithm, the pivot is chosen as the average of random
vector id's from the current subcube, and partitioning is made according
to the vector id accordingly. After the communication step of the
hypercube quicksort-like algorithm, an association list merging
algorithm combines the results so that the entire association list at
hand is sorted, and association pairs with identical vector id's are
collapsed into a single pair with accumulated scores. In the end of
the accumulation algorithm, the output assocation list is partitioned
over the processors so that the filtering of scores is also carried
out in parallel.

\subsubsection{Processing in vector blocks}

Since we process each vector separately in the basic parallelization
outlined above, although the total load of each processor is balanced,
a fine-grain imbalance is caused by the load imbalance of individual
local score calculations of a vector $x$ on each processor. To prevent
this fine-grain imbalance problem, and also decrease the latency
overhead, we process vectors not one by one, but in chunks of vectors,
so that we can use a burst-mode communication. This requires also
bundling the intermediate values so it naturally creates some
algorithmic complexity, but in practice we have seen this to be an
effective optimization for cluster architectures. Therefore we assume
that this optimization has also been made.

\subsubsection{Implementation considerations}
\label{sec:vertimpl}

There are three design options, selecting whether to implement the
local pruning optimization proposed in \prettyref{sec:localpruning}, 
selecting the score accumulation
algorithm which is either of flat  accumulation or
hypercube accumulation, and selecting how many vectors to process at each
communication step for the block processing optimization.

We have implemented all of these different options and tested them.
We implement the local pruning algorithm in the present experiments, because
it is the fastest as it reduces the number of candidates considerably
(an order of magnitude), however there are some bottlenecks in the
current implementation.  Currently, we process in blocks of $64$
vectors. Since each vector incurs memory storage for score arrays and
candidate sets, we cannot process too many vectors at once. However, a
large block size is beneficial for reducing the processing and
communication imbalance across synchronization points.

\subsection{Horizontal algorithm: partitioning vectors }
\label{sec:horiz}

The horizontal parallel algorithm partitions vectors instead of dimensions in
the vertical algorithm.  
In this parallelization, we partition the vectors and
then index and match in parallel without making much
modification to the inner loop (matching), executing matchings in
parallel over disjoint sets of vectors, however having to broadcast
each vector. 
We distribute the vectors in cyclic fashion prior to the invocation and
timing of the horizontal algorithm.


\begin{center}
  \begin{table*}
    \centering
    \caption{Real-world datasets used in our performance study.}
    \begin{tabular}{l r r r r r l}
      Dataset & $n$ & $m$ & \# non-zeroes & avg. vector size & avg. dim size  & sparsity \\ \hline
      radikal & 6883 & 136447 & 1072472 & 155.8 & 7.8 & 0.00114 \\
      20-newsgroups & 20001 & 313389 & 2984809 & 149.2 & 9.5 & 0.000476\\
      wikipedia & 70115 &1350761 & 43285850 & 617.3 & 32.0 & 0.000457 \\
      facebook & 66568 & 4618973 & 14277455 & 214.5 & 3.1 &  0.0000464 \\
      virgina-tech & 85653 & 367098 & 25827347 & 301.5 & 70.3 & 0.000821
    \end{tabular}
    \label{tab:datasets}
  \end{table*}
\end{center}

\subsubsection{Outer loop parallelization of all-pairs-0}

We now discuss how to parallelize the outer loop of \var{all-pairs-0} kind
of algorithms.
In \algo{Par-All-Pairs-0-Horiz}
(\prettyref{alg:par-all-pairs-0-horiz}), each processor is given a
disjoint set of vectors, i.e., each vector has a home processor.  Each
processor indexes only their local set of vectors; the
inverted index being constructed is partitioned horizontally, 
aligned with the input dataset partition. We pad local list
$V$ of vectors with empty vectors so that each processor has the same
number of vectors, by calculating the maximum number of vectors on a
processor with a collective communication.
For each iteration of the outer loop over local vectors, 
every processor gathers their current vector on all processors,
constructing an array of vectors $xa$ where $xa[proc]$ contains the query
vector from processor $proc$. We then iterate over all processors
$0\ldots p-1$, matching the entire set of $p$ current
query vectors against the local inverted index $I$, using the sequential matching
algorithm \algo{Find-Matches-0} of \var{all-pairs-0}. We process in the same
order on all processors to avoid redundant matches. We carefully index
the local current vector only after it has been matched against the
inverted index.

\begin{algorithm}
  \caption{$\algo{Par-All-Pairs-0-Horiz}(V,t)$}
  \label{alg:par-all-pairs-0-horiz}
  \begin{algorithmic}
    \STATE $M \gets \emptyset$
    \STATE $I \gets
    \algo{Make-Sparse-Matrix}(m, n)$ 
    \STATE Pad $V$ with empty vectors
    \FORALL{$x \in V$}
      \STATE $xa \gets \algo{MPI-All-Gather}(x)$
      \FORALL{$proc \gets 0$ to $ p-1$} 
        \STATE $M \gets M \cup 
        \algo{Find-Matches-0}(xa[proc], I, t)$ 
        \IF{$proc=pid$}
          \FORALL{$x[j] \in V$ where $x[j]>0$} 
            \STATE $I_{ji} \gets x[j]$
          \ENDFOR
        \ENDIF
      \ENDFOR
    \ENDFOR
    \RETURN $M$
  \end{algorithmic}
\end{algorithm}


 


\subsubsection{Optimizations and scalability} 

The block processing optimization may be applied to the horizontal
algorithm to improve load balance, although load balance does not
suffer much in the horizontal parallelization. Initial distribution
may be improved with respect to the random distribution balancing the
vector sizes processed in each vector iteration of the parallel
algorithm. 

Compared to \cite{plaku-distknn}, we make use of the inverted index
construction logic of \var{all-pairs-0-array}, not depending on any complex
geometric data structures, and we make use of efficient collective
communications of the message passing system, and provide a very
sensible synchronization of processing rather than having to deal with
dynamically load balancing, which results in a very elegant
algorithm. 
Nevertheless, one of their optimizations
involving bounding hyperspheres of point sets may be incorporated into
the horizontal algorithm. If the vectors are initially partitioned
geometrically, instead of cylically or according to sizes, then
bounding regions may be defined over each processor, and it may be
possible to skip some communication and computation, although we would
not expect much gain from such a computation.

The most significant parallel overhead here is the broadcast of the
vectors. Therefore, there is a total communication volume of
$size(V).(p-1)$ vector elements,
which limits scalability in high number of processors in practice.  
We have found no simple solution to this obstacle to scalability
without making a substantial re-design of the horizontal
algorithm.

\section{2-D Parallel Algorithm}
\label{sec:2dpar}

Let there be a 2-D processor mesh with $q$ rows and $r$ columns.  The
two-dimensional data partitioning algorithm combines the vertical
and horizontal parallelization as two respective levels of
parallelization. First, we make a checkerboard partitioning of the
input dataset where we distribute dimensions into $r$ columns so as to
balance load across processor columns, and we distribute vectors into
$q$ rows in cyclic order.  Therefore, to each processor, we assign a
set of vectors and a set of dimensions.
\prettyref{alg:par-all-pairs-2d} shows the pseudocode for the
\algo{Par-All-Pairs-0-2D} algorithm. We assume in the following 2-D
algorithm that $mycol$ is the processor column of the current
processor, $colid$ is the current processor's identifier in $mycol$,
$myrow$ is the processor row of the current processor, and $rowid$ is
the current processor's identifier in $myrow$.

\begin{algorithm}
  \caption{$\algo{Par-All-Pairs-0-2D}(V,t)$}
  \label{alg:par-all-pairs-2d}
  \begin{algorithmic}
    \STATE $M \gets \emptyset$ 
    \STATE $I \gets \algo{Make-Sparse-Matrix}(m, n)$ 
    \STATE Pad $V$ with empty vectors
    \FORALL{$x=v_i \in V$} 
       \STATE $xa \gets \algo{MPI-All-Gather}(x, mycol)$ 
       \FORALL{$proc \gets 0\quad$ to $q-1$} 
         \STATE $M \gets M \cup \algo{Par-Find-Matches-0-Vert}
         \linebreak (xa[proc], I, t, myrow)$ 
         \IF{$proc = colid$} 
           \FORALL{$x[j] \in V$ where $x[j]>0$}
             \STATE $I_{ji} \gets x[j]$
           \ENDFOR
         \ENDIF
      \ENDFOR
    \ENDFOR
    \RETURN $M$
  \end{algorithmic}
\end{algorithm}

The two parallelizations can be elegantly combined by re-using the
horizontal parallelization in the first level of parallelization and
the vertical parallelization in the second level. Passing the $mycol$
communicator to the vertical parallelization let us re-use the
vertical algorithm with no modification.

The block optimization of the vertical algorithm has also been tried
in our 2D experiments, but was found to cause more overhead compared
to the one that does not block input vectors.  In general, the
implementation of the vertical algorithm was found to have a
significant amount of garbage collection overhead since a lot of
intermediate data is constructed and then discarded in the vertical
algorithm (This accounts for about $30\%$ of the running time). This
overhead shows that there is room for improvement in the
implementation of our 1-D vertical and 2-D algorithms due to the
OCaml runtime overhead.

\section{Performance Study}
\label{sec:perf}

We first explain the datasets used and take a look at how the variants
of the sequential all-pairs algorithms stack up. We have based our
parallelizations on these results.   Then, we demonstrate the parallel
performance of our vertical, horizontal and 2D parallelizations. The
parallelizations do show enough diversity in performance to justify
the need for multiple parallelizations.


\subsection{Datasets}

We have based our performance study on real-world datasets, with no
tuning that will make our task easier. We have made our experiments on
two small and three large such datasets, the properties of which are
summarized in \prettyref{tab:datasets}. The columns of
\prettyref{tab:datasets} display the dataset name, number of vectors
($n$), number of dimensions ($m$), number of non-zeroes (sum of
$|x|$'s), average vector size (average of $|x|$'s), average dimension
size (average of $|I_d|$'s) and sparsity (number of non-zeroes divided
by $n.m$), respectively.

Radikal data set contains 6893 short news articles from the website of
Radikal Turkish newspaper, partitioned to 14 newspaper sections.  The
HTML documents were converted to text and converted to vector space
representation using TFIDF weighting.  20-newsgroups is a classical
text categorization dataset which consists of one thousand posts taken
from 20 USENET newsgroups.  The large datasets are downloaded from the
Stanford WebBase Project \cite{hirai2000}.  The facebook dataset is
composed of pages collected from Facebook on 09-08-2008. The wikipedia
dataset is composed of pages collected from Wikipedia on 05-2006. The
virgina-tech dataset is composed of pages collected from sites related
with Virginia Tech shooting on 04-23-2007.

\subsection{Implementation}

We have implemented the algorithms using OCaml programming language
3.12, and the OCaml MPI bindings for communication.
We have found OCaml to be effective for implementing complicated
algorithms  due to type safety and high abstraction level, while maintaining
performance.

Since OCaml does not have 32-bit floating point values, we resorted
to a fixed point implementation that uses 32 bits to store numbers and
reserves a number of fixed bits to integer and decimal point parts.
In very few cases there is some loss of accuracy which causes some
pairs to be missed but that is insignificant enough that we will not
analyze it.

We have in general paid attention to low-level issues and used fast
data structures such as arrays and lists where applicable.  For the
hash tables in candidate maps of original \var{all-pairs-0}, \var{all-pairs-1},
and \var{all-pairs-2} algorithms, we used OCaml's Hashtbl implementation in
the standard library. We initialize the hash table with one fourth of
the number of vectors.  For document vectors we used compressed row
storage, on arrays. For inverted lists we used dynamically sized
vectors with a fast way ($O(1)$) to pop the beginning of the inverted
list, which is required by the minsize optimization.

OCaml was quite suitable for implementing parallel algorithms in a
concise and reliable manner. The watertight strong typing of OCaml
helps implement algorithms with very few errors; almost all ordinary
programming errors are caught by the type checker. 
We have found that mixing functional and
imperative programming styles was quite natural in OCaml, making the
resultant code quite readable. \prettyref{fig:hypercube-mnac-all} and
\prettyref{fig:hypercube-accumulate} give a sample of real-world OCaml
code in our project. \prettyref{fig:hypercube-mnac-all} is a multi
node accumulation code that accumulates to all processors on a
hypercube network, and is the equivalent of MPI Allreduce operation;
note that the binary reduction operator is simply a
function. \prettyref{fig:hypercube-accumulate}  shows a 
fast parallelization of the score accumulation algorithm where 
each \var{key} is given a home processor $\var{key} \bmod p$.
We make p collective communications, in each of which we 
reduce to processor p pairs with the key assigned to it. The
reduction is performed by multi node accumulation to all in which
the combination operator merges two association lists into one. The
other score accumulation algorithms are also implemented at a
high-level similarly; however, since multiple optimizations are
employed there has been a moderate amount of code complexity.

\begin{figure*}
  \centering
\begin{verbatim}
let hypercube_mnac_all ?(comm=Mpi.comm_world) (a: 'a) op : 'a  =
  let p = Mpi.comm_size comm and pid = Mpi.comm_rank comm in
  let d = Util.log2_int p in
  let result = ref a in
    if debug then lprintf "d=%d\n" d;
    for dim=d-1 downto 0 do (* process dimensions of hypercube *)
      let partner = pid lxor (1 lsl dim) in
        if debug then lprintf "partner=%d\n" partner;
        let partners_result = exchange ~comm:comm !result pid partner in 
          result := op !result partners_result
    done;
    !result
\end{verbatim}
  \caption{OCaml code for multi node accumulation to all processors on a hypercube network}
  \label{fig:hypercube-mnac-all}
\end{figure*}

\begin{figure*}
  \centering
\begin{verbatim}
  let hypercube_accumulate_scores_fast al =
  let alslice = Array.make p [] in
    (*partition pairs according to cyclic-distribution of key *)
    List.iter
      (fun (key,weight)->
         let dest = key mod p in
           alslice.(dest) <- (key,weight)::alslice.(dest)
      ) al;
    let alslice = Array.map List.rev alslice in 
    let result = ref [] in
      for dest=0 to p-1 do
        let x = hypercube_mnac_all alslice.(dest) merge_als in
          if dest=pid then
            result := x
      done;
      !result
\end{verbatim}
  \caption{OCaml code for fast score accumulation using hypercube topology }
  \label{fig:hypercube-accumulate}
\end{figure*}

\subsection{Sequential performance}
\label{sec:seqperf}

\begin{table}[t]
  \centering
  \caption{Sequential running time on radikal dataset}
 \begin{tabular}{l r r r}
    Algorithm & $t=0.2$ & $t=0.3$ & $t=0.4$ \\ \hline
    all-pairs-0 &141.62 &142.34 &143.14\\ 
    \emph{all-pairs-0-array} & \emph{24.57} &\emph{24.44} &\emph{24.74}\\ 
    all-pairs-0-array2 & 29.50 &29.37& 29.53\\ 
    all-pairs-0-remscore& 180.21 &179.75& 180.41\\ 
    all-pairs-0-minsize& 149.37 &149.70 &149.43\\ 
    all-pairs-1& 87.10 & 79.05 &71.73\\ 
    all-pairs-1-array & 73.02 &69.54 & 64.40\\ 
    all-pairs-1-remscore& 180.55 & 181.79& 182.02\\ 
    all-pairs-1-upperbound &200.96 & 171.42 & 145.90\\ 
    all-pairs-1-minsize & 89.57 & 80.52 &72.21\\ 
    all-pairs-1-remscore-minsize & 93.31 & 82.70 & 73.42\\ 
    all-pairs-2& 198.92 & 165.64 & 138.98 \\
    all-pairs-bruteforce & 183.06 &183.32 &183.28 \\ \hline
  \end{tabular}
  \label{tab:radikalseqtime}
\end{table}

\begin{table}
  \centering
  \caption{Sequential running time on 20-newsgroups dataset }
  \begin{tabular}{l r r r}
    Algorithm & $t=0.4$& $t=0.5$& $t=0.6$\\ \hline
    all-pairs-0 &2887.1 &2904.7 &2900.9\\ 
    \emph{all-pairs-0-array} &\emph{480.0} &\emph{495.4} &\emph{478.0}\\ 
    all-pairs-0-array2 & 596.4 & 618.0 & 595.6\\ 
    all-pairs-0-remscore & 3482.2 & 3486.1 & 3501.2\\ 
    all-pairs-0-minsize&3171.6 & 3180.7 & 3191.4\\ 
    all-pairs-1 &1169.0 & 1035.6 & 882.9\\ 
    all-pairs-1-array & 883.8 & 837.0 &757.5\\ 
    all-pairs-1-remscore&3502.0 & 3499.9 & 3497.7\\ 
    all-pairs-1-upperbound&1940.8 &1649.5 & 1410.1\\ 
    all-pairs-1-minsize&1190.0 & 1023.5 & 902.6\\ 
    all-pairs-1-remscore-minsize & 1288.5 & 1086.4 & 943.6\\ 
    all-pairs-2 &1733.8 & 1450 & 1190.2\\ 
    all-pairs-bruteforce & 1866.3 & 1867.2 & 1871.5 \\ \hline
  \end{tabular}
  \label{tab:20-newsgroupsseqtime}
\end{table}

Table \ref{tab:radikalseqtime} shows the running times of various
sequential all-pairs algorithms on radikal dataset with a few
meaningful support thresholds.  Likewise, Table
\ref{tab:20-newsgroupsseqtime} shows the running time on the
20-newsgroups dataset. The dot-product thresholds were chosen so that
we obtain roughly $n.\lg(n)$ pairs for $n$ vectors and increase the
threshold until we have about $n$ pairs. $n.\lg(n)$ pairs should be
sufficient to construct a well connected epsilon neighborhood graph,
given each vector has about $\lg(n)$ neighbors, since it is well-known
that to establish inter-cluster connectivity setting $k \sim \log(n)$
is the lowest sufficient rate for knn graphs \cite{brito1997}.

We have had to compare the effects of different optimization
strategies so that we could determine which algorithm could be
parallelized best. In case, there is a clearly best algorithm, the
other parallelizations would be redundant. Not all optimization
strategies may be parallelized well, either. The effect of
optimizations may also depend on the software and hardware platform
used. Therefore, we made variations on the original \var{all-pairs-0}
algorithm so that optimizations were tested individually and in
combination. Surprisingly, it turns out that the best algorithm is an
optimization of \var{all-pairs-0} itself that uses arrays instead of hash
tables for score accumulation, named \var{all-pairs-0-array} in the tables.
The running times show that it is actually quite difficult to improve
upon the brute force algorithm. While that may sound frustrating, it
also means that there is positive research potential in designing new
all pairs similarity algorithms, since it is known that the optimal
nearest neighbor query algorithms in $R^d$ have low running time
complexity \cite{callahan1995}. With some work, those results could
carry over to real-world data with high dimensionality and sparse
vectors.

Algorithm \var{all-pairs-0-array2} fares worse than \var{all-pairs-0-array}
unexpectedly, the only difference in the former is maintaining a list
so that the written entries can be zeroed out in the next
iteration. It seems that this list maintenance and zeroing only
written entries is more expensive than zeroing out the entire array,
suggesting that the non-zero entries are too many for this
optimization. We also see that, somewhat unexpectedly, \var{all-pairs-2}
does not improve on \var{all-pairs-1}. However, the minsize optimization
over \var{all-pairs-1} improves on \var{all-pairs-1}, which is part of
\var{all-pairs-2}. The other two optimizations of \var{all-pairs-2}
(\var{all-pairs-1-remscore} and \var{all-pairs-1-upperbound}) apparently slow down
\var{all-pairs-1} instead of accelerating it. \var{All-pairs-1-remscore-minsize}
is also worse than \var{all-pairs-1-minscore}, suggesting that the remscore
optimization is not useful at all. \var{All-pairs-2} is almost as slow as
bruteforce for lower thresholds, so it may not be a very meaningful
algorithm to study. Even at higher thresholds, where there are too few
outputs, the results do not change significantly.  Still, the most
interesting result is that, \var{all-pairs-0-array} is much better than all
of those optimizations. The array optimization carries over to
\var{all-pairs-1}, but \var{all-pairs-1-array} still does not match
\var{all-pairs-0-array}, so we did not feel obliged to apply it to
other variants. Likewise, stand-alone optimizations over \var{all-pairs-0}
(\var{all-pairs-0-remscore}, and \var{all-pairs-0-minsize}) seem to be worse than
the brute force algorithm therefore they were not worth pursuing.

While in this work, we focus on the parallelization of existing
algorithms and their variants, we have also examined the real reason
for the high running times. The first reason is that there are a lot
of dimensions, breaking down easy separability of points, and second
and most importantly that the density of the dimensions follow a
power-law distribution which introduces an almost irreducible
complexity in the processing of the densest dimensions. The reason why
partial indexing optimization is more effective than the other
optimizations is that it separates the processing into a dense and a
sparse phase, where a brute force algorithm is applied to the dense
part of the data and an indexing approach is applied to the sparse
part, definitely improving over the plain brute force algorithm
especially in the case of higher thresholds as can be seen in the
running times. Still, the improvement seems to be on a constant order,
which is interesting as it suggests that there is no asymptotic
improvement.

\subsection{Parallel performance}
\label{sec:speedup}

Our experiments were carried out at the TUBITAK ULAKBIM High
Performance Computing Center, which is comprised of $48$-core
multi-processor nodes built with AMD Opteron 6172 processors,
interconnected with an Infiniband network, running GNU/Linux operating
system. In the rest of the paper, we use processor and core
interchangeably.  For each dataset, we worked on a single, meaningful
similarity threshold.  The similarity thresholds for each dataset were
chosen so that they would result in a well-connected similarity
graph. We again followed the notion of allowing about $n \log n$
similar pairs in the output as a rough guideline, we made sure that we
obtained a significant number of similar
pairs. \prettyref{tab:problems} shows the problem instances used in
our parallel performance study; the columns display the similarity
threshold $t$, the running time of the sequential algorithm
\var{all-pairs-0-array} and the number of similar pairs output.  For 1-D
algorithms, we ran our algorithms up to 256 processors on the small
datasets and up to 64 processors on the large datasets due to resource
limitations on the batch system of the supercomputer (128 in one
instance where we could run it). For the 2-D algorithms, we have tried
different combinations of numbers of processor rows and columns in the
virtual mesh, again going up to 256 at most for small datasets and 128
processors at most for large datasets.

\begin{table}
  \centering
  \caption{The problem instances used in our study.}
  \begin{tabular}{l r r r r r l}
    Dataset &  $t$ & Time & Matches \\ \hline
    radikal & 0.2 & 15.5  & 16810 \\
    20-newsgroups & 0.4 & 317.3  & 64396 \\
    wikipedia & 0.9  & 54424.0 & 747999 \\
    facebook & 0.99 & 10777.8 & 819196\\
    virgina-tech & 0.99 & 10426.2 & 13447874
  \end{tabular}
  \label{tab:problems}
\end{table}

We have performed our experiments on three algorithms. The vertical
parallel algorithm is \prettyref{alg:par-all-pairs-0-vert} that uses
the major optimization of score accumulation with local pruning
(\prettyref{lem:localpruning}), since it provides the best results.
 We use the flat score accumulation algorithm in the experiments
(\prettyref{sec:flataccum}), the other choices were covered in
\prettyref{sec:vertimpl}. The horizontal parallel algorithm is
\prettyref{alg:par-all-pairs-0-horiz}, explained in
\prettyref{sec:horiz}.

\prettyref{fig:speedup-small} shows the parallel speedup of our
vertical and horizontal algorithms on the smaller two datasets:
radikal and 20-newsgroups. Likewise, \prettyref{fig:2dspeedup-small}
shows the speedup of our 2-D algorithm on the same datasets.
Similarly, \prettyref{fig:speedup-large} depicts the parallel speedup
of our 1-D algorithms on the larger three datasets: wikipedia,
facebook and virgina-tech, while \prettyref{fig:2dspeedup-large} gives
the speedups of the 2-D algorithm for the same three datasets.  The
processor configurations of the 2D algorithm are indicated as $p
\times q$ on the $x$-axis where $p$ is the number of the processor
rows and $q$ is the number of processor columns. Note that the
vertical algorithm was run on up to $32$ processors on small datasets,
and up to $16$ processors for large datasets, beyond which the
algorithm becomes infeasible to run. A few of the results are missing
in the 2-D algorithm plots in \prettyref{fig:2dspeedup-large} due to
unresolvable system problems that we encountered on the shared
supersomputer that we used. The system consistently stalled when
we submitted some large parallel jobs, possibly due to a bug in the
interconnection network. It should be clear from the figures that the
few missing data points do not change the general picture.

\begin{figure*}[htbp]
  \begin{center}
    \subfigure[radikal]{
      \includegraphics[scale=0.435]{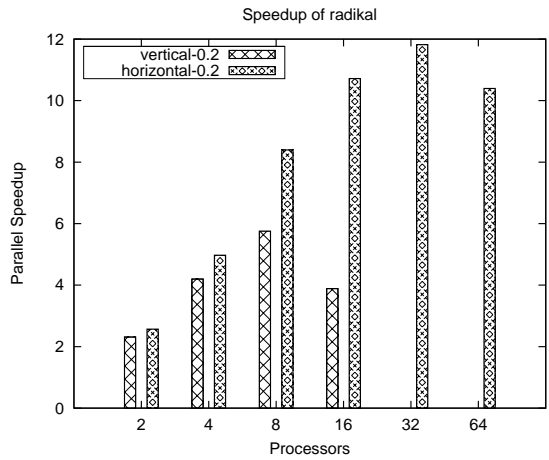}
    } \subfigure[20-newsgroups]{
      \includegraphics[scale=0.435]{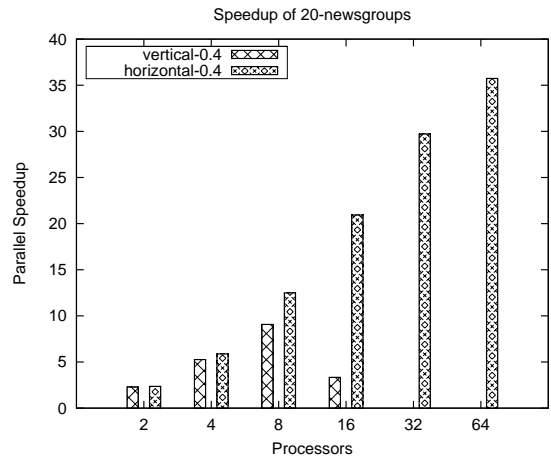}
    }
    \caption{Parallel speedup of horizontal and vertical algorithms on
      small datasets radikal and 20-newsgroups}
    \label{fig:speedup-small}
  \end{center}
\end{figure*}

\begin{figure*}[htbp]
  \begin{center}
    \subfigure[radikal]{
      \includegraphics[scale=0.435]{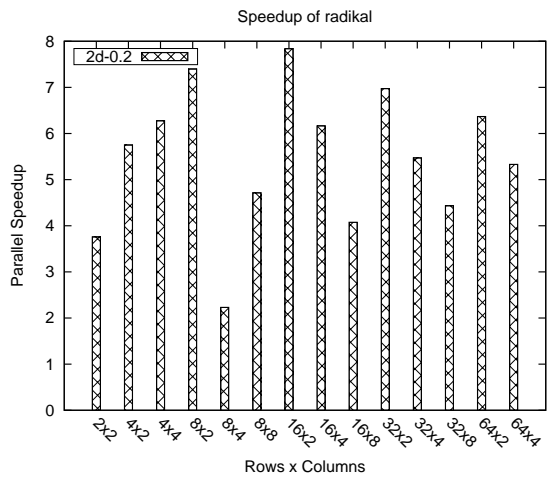}
    } \subfigure[20-newsgroups]{
      \includegraphics[scale=0.435]{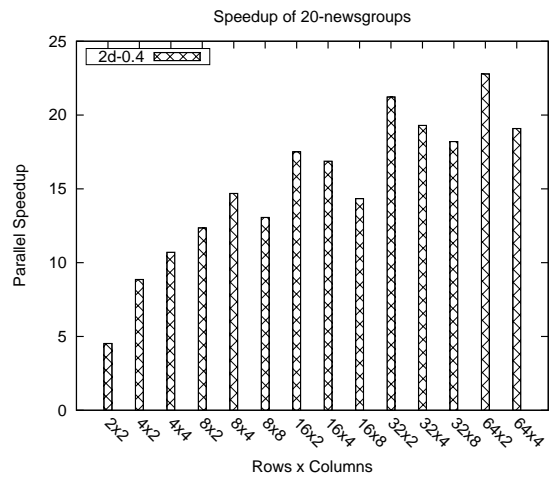}
    }
    \caption{Parallel speedup of the 2D algorithm on small datasets
      radikal and 20-newsgroups}
    \label{fig:2dspeedup-small}
  \end{center}
\end{figure*}

\begin{figure}[!]
  \begin{center}
    \subfigure[wikipedia]{
      \includegraphics[scale=0.435]{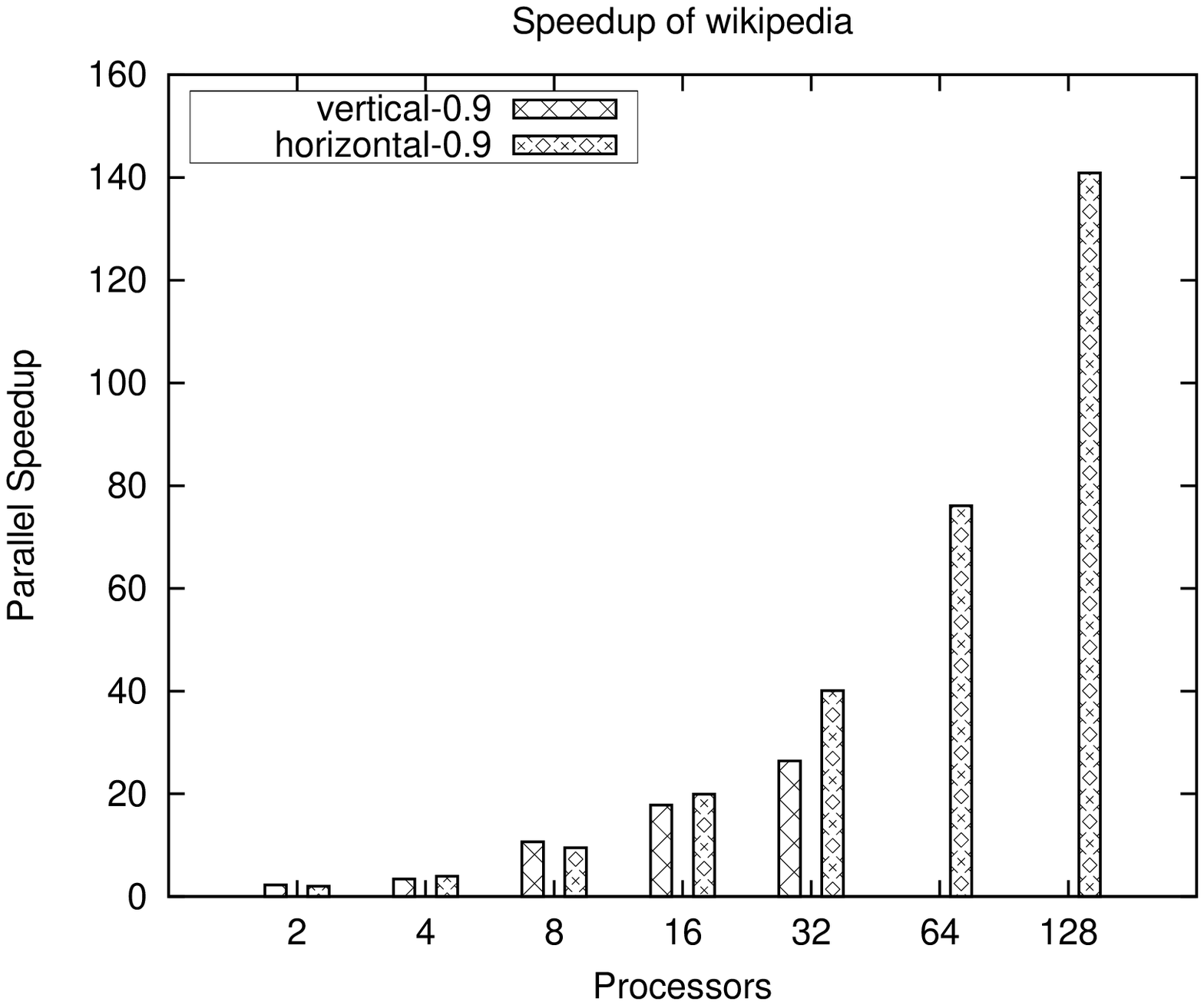}
    } \subfigure[facebook]{
      \includegraphics[scale=0.435]{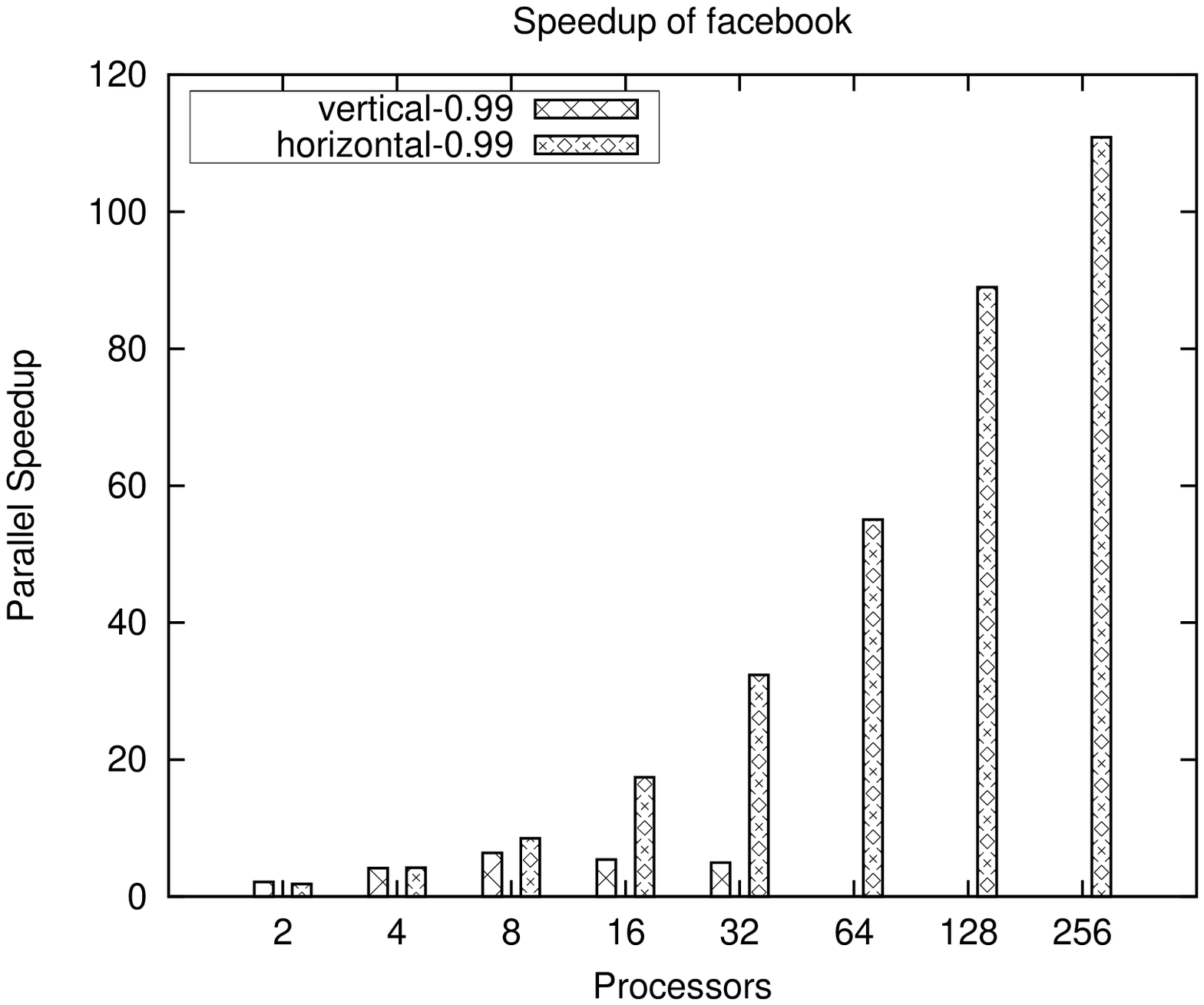}
    } \subfigure[virginia-tech]{
      \includegraphics[scale=0.435]{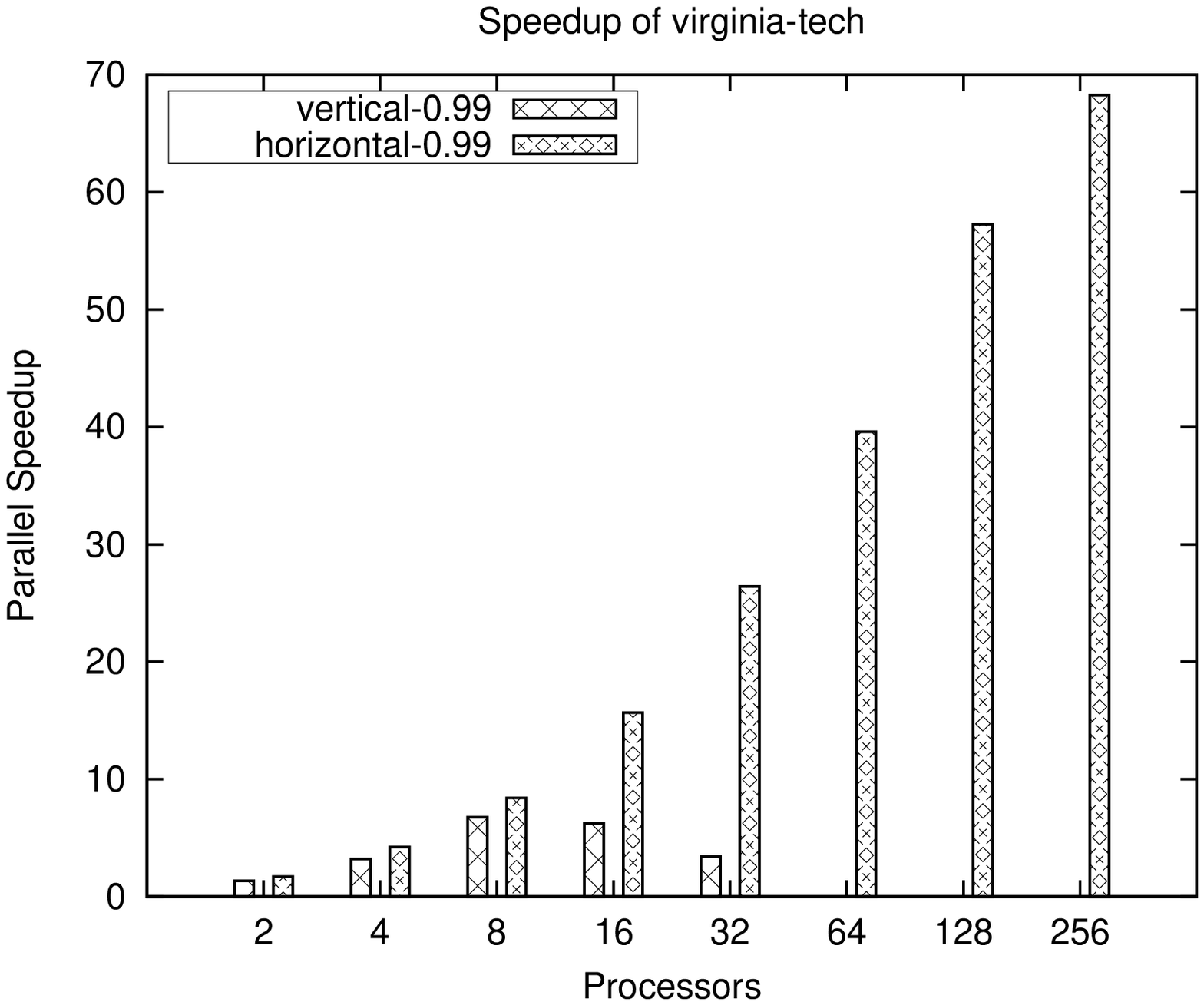}
    }
    \caption{Parallel speedup of horizontal and vertical algorithms on
      the large datasets: wikipedia, facebook, virginia-tech}
    \label{fig:speedup-large}
  \end{center}
\end{figure}

\begin{figure}[!]
  \begin{center}
    \subfigure[wikipedia]{
      \includegraphics[scale=0.435]{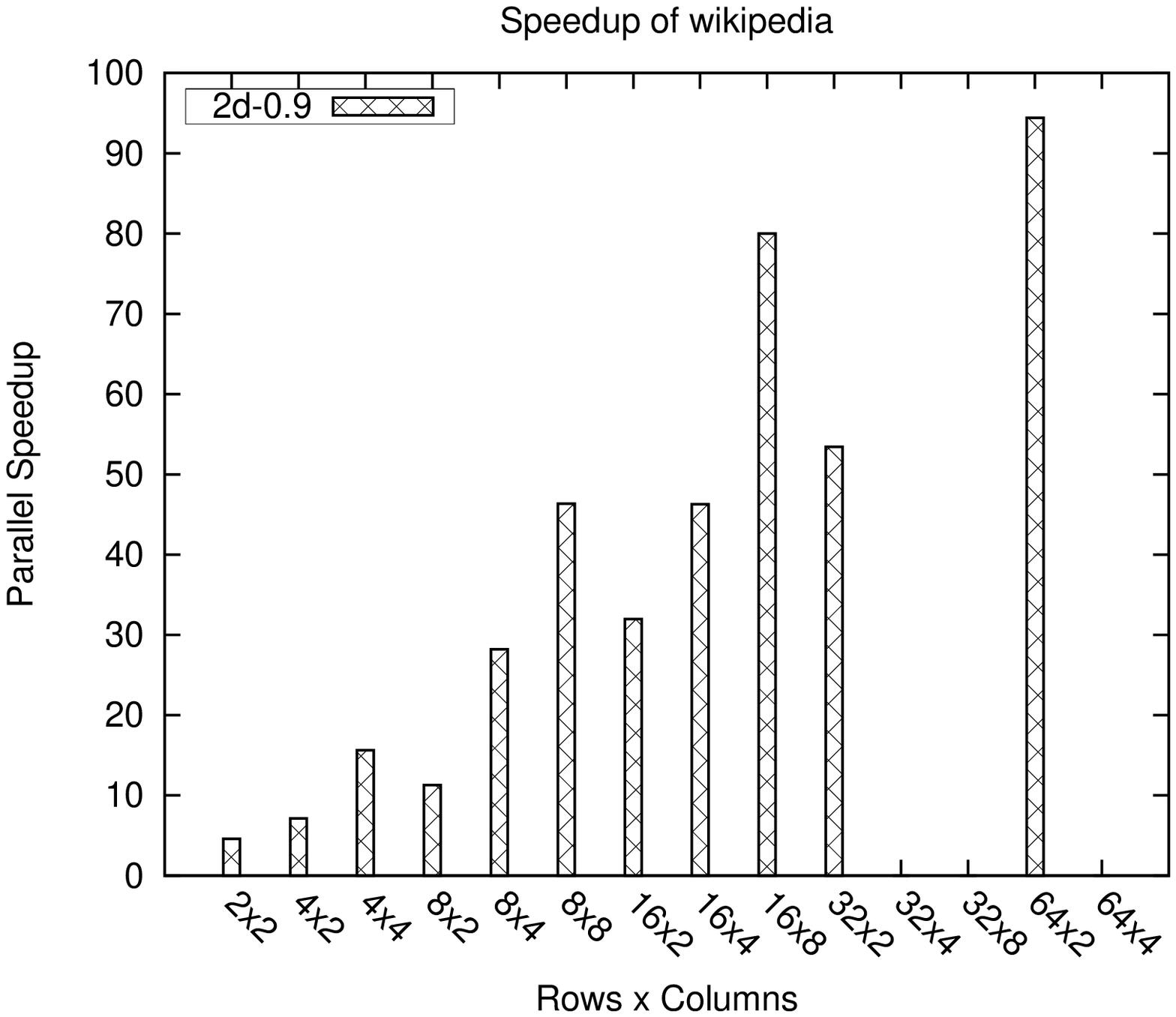}
    } \subfigure[facebook]{
      \includegraphics[scale=0.435]{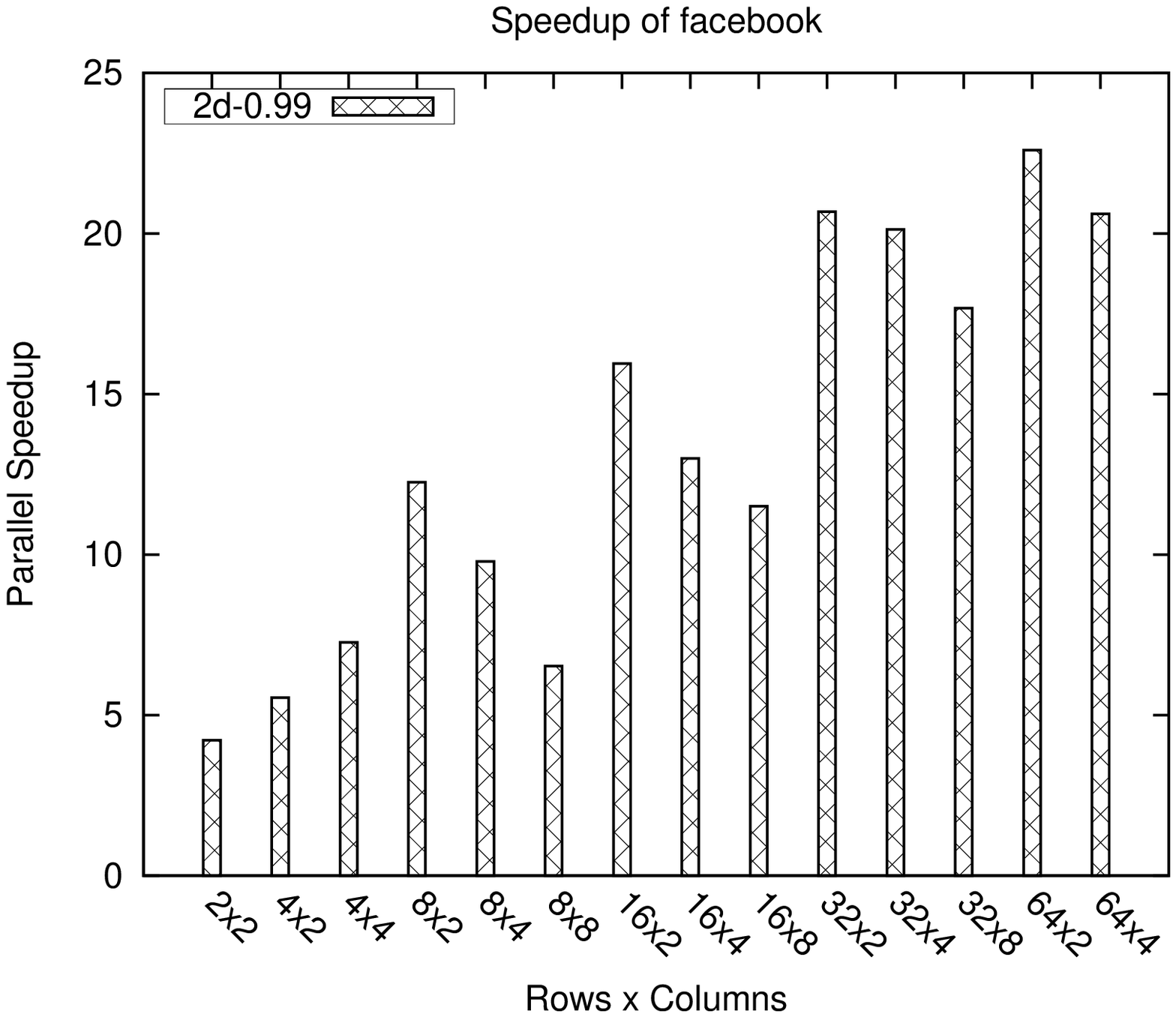}
    } \subfigure[virginia-tech]{
      \includegraphics[scale=0.435]{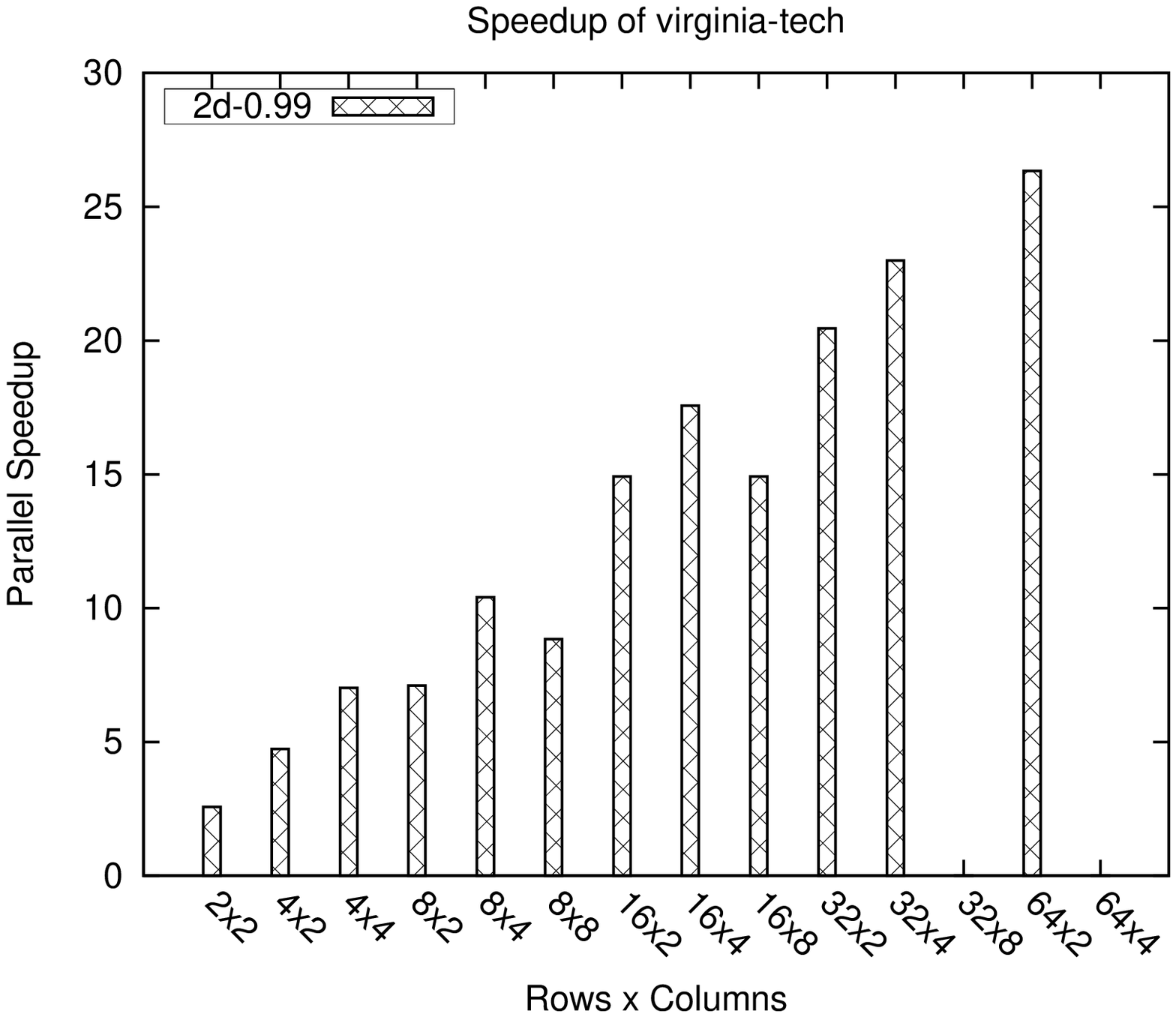}
    }
    \caption{Parallel speedup of the 2D algorithm on the large
      datasets: wikipedia, facebook, virginia-tech}
    \label{fig:2dspeedup-large}
  \end{center}
\end{figure}

As seen in \prettyref{fig:speedup-small} and
\prettyref{fig:speedup-large},
we see that the horizontal algorithm scales better
than the vertical algorithm for all datasets. The vertical algorithm scales well up to
$8$ processors, but after that it loses quite a bit of steam.  It is
still quite an achievement that the vertical algorithm scales as much,
since the number of processors increase the communication volume and
communication asynchrony rapidly despite the local pruning
optimization. The horizontal algorithm scales well up to $32$
processors and then starts to slow down due to the fact that the
broadcast starts becoming significant. This is most apparent in
radikal dataset, but it is also seen in other datasets that the
speedup does not accelerate as much, as we go up to $64$ processors.
We observe that both vertical and horizontal parallelizations achieve
super-linear speedups in several cases, affirming the efficiency of
our implementation, as in those cases the algorithms make better use
of the memory hierarchy. In two cases, we see that the vertical
algorithm achieves better speedup than the horizontal algorithm,
justifying the usefulness of our vertical algorithm. 

The 2-D algorithm
shows varying performance according to the processor
configuration as seen in  \prettyref{fig:2dspeedup-small} and \prettyref{fig:2dspeedup-large} . Since the vertical algorithm did not scale further than
$8$ processors, we did not try more processor columns in the virtual
mesh. We sometimes see excellent speedups with the 2-D algorithm, for
instance in wikipedia, $4 \times 4$ yields super-linear speedup and
$16 \times 8$ yields about $80$ speedup. However, on the average, the
2-D algorithm's performance is between that of the horizontal and
vertical algorithms, it is usually better than half of the speedup of
the horizontal algorithm for the maximum number of processors although
for facebook dataset it's slightly worse than that.

\subsection{Local pruning and block processing optimizations}

It is useful to understand the performance impact of local pruning and
block processing optimizations for the vertical algorithm. Without
those optimizations, the vertical algorithm is futile, it would not be
quite possible to apply it to sufficiently many cases. Therefore, we
show its performance, when neither optimization is applied, and when
only local pruning is applied, together with happens when both
optimizations are turned on.

We have chosen the smaller two datasets radikal and 20-newsgroups for
this comparison, because some of the runs would be infeasible for the
larger datasets. We run them on up to $16$ processors, which is
sufficient to illustrate the performance differences.
\prettyref{fig:algocomparison} shows how speedup varies for different
vertical algorithms on small datasets, comparing the unoptimized
vertical algorithm (vertical-noopt), the vertical algorithm with local
pruning optimization only (vertical-localpruning) and the vertical
algorithm with both local pruning and block processing optimizations
applied (vertical-bothopt). It is clearly seen that local pruning
improves over no optimization and both optimizations together improve
on local pruning only. Local pruning is more significant for smaller
number of processors and block processing is more significant for
larger number of processors. In fact, without these optimizations, we
see that the speedups would be too low. It is only due to these effective
optimizations that we have been able to obtain the speedups previously
demonstrated. The comparison is similar across two datasets. The
optimizations are most effective on $8$ processors; on $16$
processors, the effectiveness of the local pruning algorithm declines
greatly, which is why we did not extend the study to a larger number
of processors. 

\prettyref{fig:bsizecomparison} shows the speedups for
various block sizes on small datasets. We have used block sizes of
$1$, $4$, $8$, $16$, $32$, $64$ on up to $16$ processors.
We have observed that increasing the
block size does improve the speedup, especially on larger number of
processors. Speedup generally improves on
both datasets until $32$ block size, and on most until $64$ block
size (except one data point); 
though we also observe that the gains start diminishing at $64$,
which is why we stopped there. Also, 
larger block sizes turned out to be infeasible for large datasets.

Comparing speedups alone does not give us much insight into how these
speed differences occur. We have thus profiled the algorithms in
detail. We have measured the time elapsed for both communication and
computation phases in the algorithms. We have also calculated how the
number of candidates vary when we use the optimizations, and how many
scores are actually accumulated. We have also put a barrier before
each collective communication operation, so that we can measure how
much processors wait before engaging in actual communication. We give
both average and maximum values for the measured values, to show how
imbalance for these values vary. In \prettyref{tab:radikal-variants},
\prettyref{tab:20-newsgroups-variants},
\prettyref{tab:radikal-blocksize}, and
\prettyref{tab:20-newsgroups-blocksize} we measure the following
parameters for varying number of processors and algorithms: p shows
the number of processors, Algo. shows the algorithm being used,
$C_{avg}$ shows the average time of communication, $C_{max}$ shows the
maximum time of communication, $W_{avg}$ shows the average time of
work, $W_{max}$ shows the maximum time of work, Scores shows the total
number of scores communicated, $Cand_{avg}$ shows the average number
of candidates, $Cand_{max}$ shows the maximum number of candidates
$Barr_{avg}$ shows the average barrier time, $Barr_{max}$ shows the
maximum barrier time.

\prettyref{tab:radikal-variants} and
\prettyref{tab:20-newsgroups-variants} show the profiling results for
the three vertical algorithm variants previously mentioned.  The
profiling data suggests that the local pruning optimization is
effective for reducing communication time, and the number of scores
communicated. On $2$ processors, we see that it reduces more than
100-fold.  Even on $16$ processors, there is a 10-fold improvement on
the number of scores communicated. The work time is also reduced due
to fewer scores being processed. The barrier time also reduces
favorably for local pruning optimization. However, block processing
further reduces barrier time, and consequently, the communication
time. It turns out that block processing optimization is very
effective for the all-pairs similarity problem, as otherwise the
effects of small communication latencies and imbalances must be
aggregating. We see that the work time slightly increases, but this is
offset by the huge savings in communication time. For instance, on $8$
processors, for 20-newsgroups dataset, the maximum communication time
reduces from $31.22$ seconds to $7.02$ seconds, while maximum work
time increases from $24.34$ seconds to $25.89$ seconds, and the
barrier time reduces from $15.87$ to $5.73$. These are quite
significant savings for a parallel algorithm.

\prettyref{tab:radikal-blocksize} and
\prettyref{tab:20-newsgroups-blocksize} show the profiling results
when only the processing block size is varied in the fully optimized
vertical algorithm, where algorithm ``vertical-bsx'' means a block
size of x.  Note that the numbers of scores and candidates do not
change in this table. We see that, generally, enlarging block size improves
reduction of communication time and barrier time. The communication
imbalances also follow a decreasing trend as the block size increases,
which shows that our statistical reasoning works. Especially, the
communication times become much more even as the block size is
increased. The barrier time also follows a similar trend, but it does
not become as finely balanced.  The communication and barrier times
are very small already even for a block size of 64, so more
intelligent document partitioning methods may not be very effective in
improving communication performance.  We did not increase the block
size much further, since every document in the block incurs a large
memory penalty. We did get out of memory errors with a block size of
64 on larger datasets. In general, the block size must be specified
with the dataset size in mind so as to prevent such errors.

\begin{figure*}
  \begin{center}
    \subfigure[radikal]{
      \includegraphics[scale=0.46]{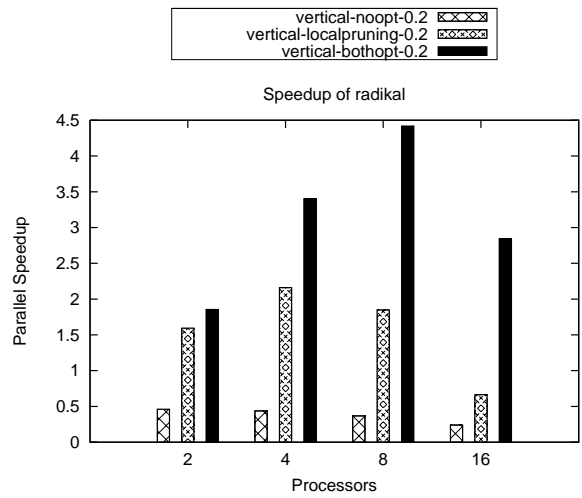}
    } \subfigure[20-newsgroups]{
      \includegraphics[scale=0.46]{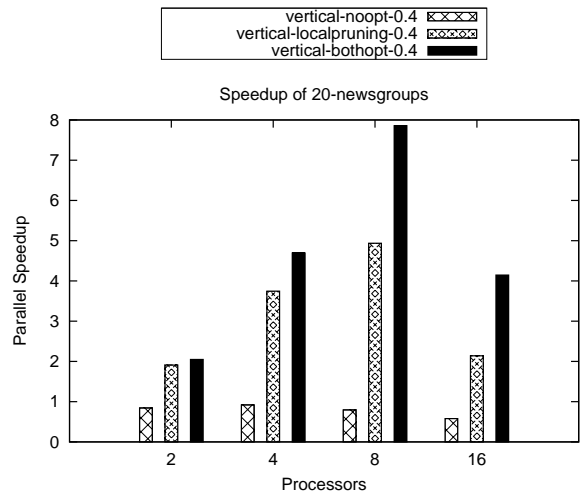}
    }
    \caption{Speedup comparison of three parallel algorithms on
      radikal and 20-newsgroups datasets}
    \label{fig:algocomparison}
  \end{center}
\end{figure*}

\begin{figure*}
  \begin{center}
    \subfigure[radikal]{
      \includegraphics[scale=0.46]{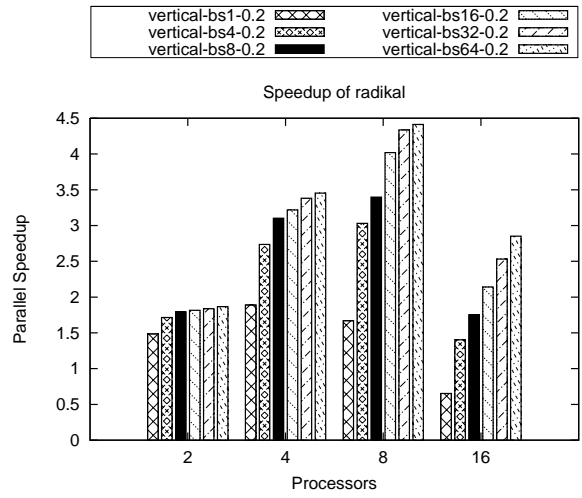}
    } \subfigure[20-newsgroups]{
      \includegraphics[scale=0.46]{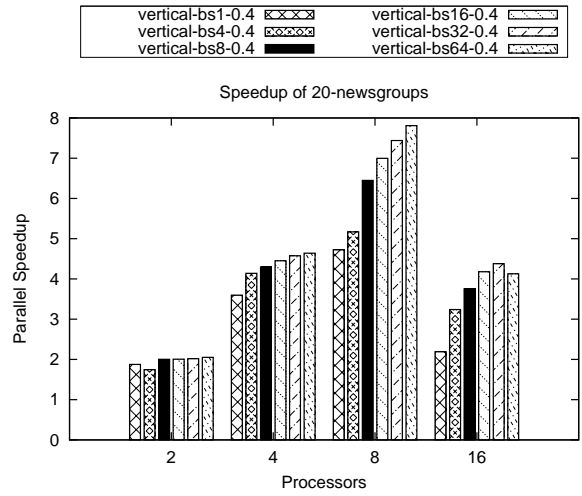}
    }
    \caption{Speedup comparison of varying block sizes on radikal and
      20-newsgroups datasets}
    \label{fig:bsizecomparison}
  \end{center}
\end{figure*}
\begin{table*}
\centering
\caption{Profiling of vertical variants on radikal dataset}
 \begin{tabular}{r l r r r r r r r r r}
p & Algo. & $C_{avg}$ & $C_{max}$ & $W_{avg}$ & $W_{max}$ & $Barr_{avg}$ & $Barr_{max}$ & Scores & $Cand_{avg}$ & $Cand_{max}$ \\
\hline 
\hline \\
2 &vertical-noopt & 12.03 & 12.09 & 8.63 & 8.71 & 3.15 & 4.42 & 23684403 & 0.0 & 0\\
2 &vertical-localpruning & 1.27 & 1.32 & 5.81 & 5.86 & 0.74 & 0.81 & 42086 & 22886.5 & 23272\\
2 &vertical-bothopt & 0.04 & 0.04 & 6.26 & 6.36 & 0.24 & 0.34 & 42086 & 22886.5 & 23272\\
 
\hline \\
4 &vertical-noopt & 18.41 & 18.72 & 4.12 & 4.28 & 6.69 & 9.98 & 23684403 & 0.0 & 0\\
4 &vertical-localpruning & 2.34 & 2.38 & 2.78 & 2.87 & 0.97 & 1.04 & 116000 & 34393.8 & 38986\\
4 &vertical-bothopt & 0.10 & 0.12 & 3.17 & 3.25 & 0.25 & 0.27 & 116000 & 34393.8 & 38986\\
 
\hline \\
8 &vertical-noopt & 27.15 & 27.91 & 2.02 & 2.16 & 11.47 & 17.21 & 23684403 & 0.0 & 0\\
8 &vertical-localpruning & 4.55 & 4.60 & 1.55 & 1.65 & 1.51 & 1.93 & 355937 & 53711.8 & 73642\\
8 &vertical-bothopt & 0.41 & 0.51 & 1.87 & 2.02 & 0.42 & 0.60 & 355937 & 53711.8 & 73642\\
 
\hline \\
16 &vertical-noopt & 47.35 & 48.04 & 1.21 & 1.55 & 10.57 & 13.52 & 23684403 & 0.0 & 0\\
16 &vertical-localpruning & 17.07 & 17.36 & 0.93 & 1.06 & 1.69 & 2.90 & 1155714 & 89717.0 & 202112\\
16 &vertical-bothopt & 2.42 & 2.57 & 1.23 & 1.35 & 0.54 & 0.89 & 1155714 & 89717.0 & 202112\\

\end{tabular}
\label{tab:radikal-variants}
\end{table*}

\begin{table*}
\centering
\caption{Profiling of vertical variants on 20-newsgroups dataset}
 \begin{tabular}{r l r r r r r r r r r}
p & Algo. & $C_{avg}$ & $C_{max}$ & $W_{avg}$ & $W_{max}$ & $Barr_{avg}$ & $Barr_{max}$ & Scores & $Cand_{avg}$ & $Cand_{max}$ \\
\hline 
\hline \\
2 &vertical-noopt & 123.08 & 124.63 & 153.79 & 155.23 & 34.17 & 44.30 & 194138198 & 0.0 & 0\\
2 &vertical-localpruning & 12.84 & 14.91 & 134.68 & 136.86 & 10.84 & 12.84 & 287786 & 148376.0 & 246016\\
2 &vertical-bothopt & 0.17 & 0.18 & 137.88 & 139.72 & 3.29 & 5.13 & 287786 & 148376.0 & 246016\\
 
\hline \\
4 &vertical-noopt & 177.56 & 178.94 & 78.42 & 79.00 & 70.60 & 99.58 & 188179681 & 0.0 & 0\\
4 &vertical-localpruning & 18.67 & 19.49 & 54.76 & 56.04 & 14.04 & 14.55 & 1060564 & 274885.0 & 398405\\
4 &vertical-bothopt & 0.81 & 1.04 & 56.62 & 57.79 & 4.03 & 4.79 & 1060564 & 274885.0 & 398405\\
 
\hline \\
8 &vertical-noopt & 266.82 & 274.48 & 42.87 & 62.64 & 114.28 & 158.89 & 180315935 & 0.0 & 0\\
8 &vertical-localpruning & 30.78 & 31.22 & 23.33 & 24.34 & 13.79 & 15.87 & 4165217 & 551323.0 & 1939290\\
8 &vertical-bothopt & 6.19 & 7.02 & 25.22 & 25.89 & 4.83 & 5.73 & 4165217 & 551323.0 & 1939290\\
 
\hline \\
16 &vertical-noopt & 434.75 & 440.22 & 22.93 & 23.98 & 120.81 & 144.17 & 172874767 & 0.0 & 0\\
16 &vertical-localpruning & 111.18 & 111.83 & 16.82 & 18.08 & 13.58 & 17.53 & 17454734 & 1203360.0 & 11294606\\
16 &vertical-bothopt & 47.05 & 48.16 & 15.53 & 19.16 & 5.40 & 6.87 & 17454734 & 1203360.0 & 11294606\\

\end{tabular}
\label{tab:20-newsgroups-variants}
\end{table*}
\begin{table*}
\centering
\caption{Profiling of various block sizes on radikal dataset}
 \begin{tabular}{r l r r r r r r r r r}
p & Algo. & $C_{avg}$ & $C_{max}$ & $W_{avg}$ & $W_{max}$ & $Barr_{avg}$ & $Barr_{max}$ & Scores & $Cand_{avg}$ & $Cand_{max}$ \\
\hline 
\hline \\
2 &vertical-bs1 &0.73 & 0.74 & 6.48 & 6.61 & 0.90 & 0.99 & 42086 & 22886.5 & 23272\\
2 &vertical-bs4 &0.23 & 0.23 & 6.36 & 6.47 & 0.52 & 0.61 & -- & -- & --\\
2 &vertical-bs8 &0.12 & 0.13 & 6.35 & 6.38 & 0.34 & 0.36 & --& --& --\\
2 &vertical-bs16 &0.08 & 0.08 & 6.33 & 6.45 & 0.34 & 0.45 & -- & -- &--\\
2 &vertical-bs32 &0.05 & 0.05 & 6.32 & 6.34 & 0.29 & 0.31 & --& -- & --\\
2 &vertical-bs64 &0.04 & 0.04 & 6.29 & 6.33 & 0.23 & 0.25 &-- &-- & --\\
 
\hline \\
4 &vertical-bs1 &1.78 & 1.87 & 3.40 & 3.48 & 1.34 & 1.43 & 116000 & 34393.8 & 38986\\
4 &vertical-bs4 &0.50 & 0.52 & 3.29 & 3.36 & 0.67 & 0.76 &-- & -- &--\\
4 &vertical-bs8 &0.28 & 0.30 & 3.22 & 3.29 & 0.42 & 0.47 & -- & -- & --\\
4 &vertical-bs16 &0.20 & 0.23 & 3.20 & 3.29 & 0.36 & 0.41 & -- & -- & --\\
4 &vertical-bs32 &0.14 & 0.16 & 3.17 & 3.26 & 0.27 & 0.31 & -- &-- & --\\
4 &vertical-bs64 &0.10 & 0.11 & 3.16 & 3.26 & 0.24 & 0.28 & -- & -- & --\\
 
\hline \\
8 &vertical-bs1 &3.61 & 3.91 & 2.02 & 2.16 & 1.94 & 2.27 & 355937 & 53711.8 & 73642\\
8 &vertical-bs4 &1.10 & 1.24 & 1.94 & 2.06 & 1.03 & 1.21 & -- & -- & --\\
8 &vertical-bs8 &0.73 & 0.87 & 1.92 & 2.07 & 0.94 & 1.13 & -- & -- & --\\
8 &vertical-bs16 &0.50 & 0.64 & 1.90 & 2.04 & 0.60 & 0.81 & -- & -- & --\\
8 &vertical-bs32 &0.40 & 0.51 & 1.90 & 2.03 & 0.47 & 0.65 & -- & -- & --\\
8 &vertical-bs64 &0.41 & 0.52 & 1.88 & 2.03 & 0.42 & 0.60 & -- & -- & --\\
 
\hline \\
16 &vertical-bs1 &15.39 & 16.12 & 1.37 & 1.53 & 2.73 & 3.71 & 1155714 & 89717.0 & 202112\\
16 &vertical-bs4 &6.39 & 6.66 & 1.31 & 1.45 & 1.17 & 1.69 & -- & -- & --\\
16 &vertical-bs8 &4.77 & 4.95 & 1.29 & 1.42 & 0.97 & 1.40 & -- & -- & --\\
16 &vertical-bs16 &3.57 & 3.74 & 1.28 & 1.43 & 0.85 & 1.25 & -- & --& --\\
16 &vertical-bs32 &2.85 & 3.01 & 1.28 & 1.44 & 0.64 & 1.00 & -- & -- & --\\
16 &vertical-bs64 &2.43 & 2.58 & 1.23 & 1.34 & 0.55 & 0.89 & -- & -- & --\\

\end{tabular}
\label{tab:radikal-blocksize}
\end{table*}
\begin{table*}
\centering
\caption{Profiling of various block sizes on 20-newsgroups dataset}
 \begin{tabular}{r l r r r r r r r r r}
p & Algo. & $C_{avg}$ & $C_{max}$ & $W_{avg}$ & $W_{max}$ & $Barr_{avg}$ & $Barr_{max}$ & Scores & $Cand_{avg}$ & $Cand_{max}$ \\
\hline 
\hline \\
2 &vertical-bs1 &2.66 & 2.69 & 138.96 & 140.97 & 11.56 & 13.88 & 287786 & 148376.0 & 246016\\
2 &vertical-bs4 &0.88 & 0.88 & 148.88 & 158.13 & 15.14 & 24.28 & -- & -- & --\\
2 &vertical-bs8 &0.50 & 0.52 & 137.95 & 140.12 & 5.35 & 7.53 & -- & --& --\\
2 &vertical-bs16 &0.30 & 0.31 & 138.80 & 140.88 & 4.54 & 6.64 & -- & -- & --\\
2 &vertical-bs32 &0.39 & 0.40 & 138.46 & 140.61 & 3.96 & 6.10 & -- & -- & --\\
2 &vertical-bs64 &0.17 & 0.18 & 137.19 & 138.65 & 3.09 & 4.54 & -- & -- & --\\
 
\hline \\
4 &vertical-bs1 &5.87 & 6.25 & 58.95 & 59.57 & 15.29 & 16.31 & 1060564 & 274885.0 & 398405\\
4 &vertical-bs4 &2.00 & 2.14 & 58.06 & 58.83 & 9.44 & 10.42 & -- & -- & --\\
4 &vertical-bs8 &1.51 & 1.76 & 57.70 & 59.07 & 7.69 & 8.60 & -- & -- & --\\
4 &vertical-bs16 &0.92 & 1.01 & 57.32 & 57.89 & 6.27 & 7.22 & -- & -- & --\\
4 &vertical-bs32 &0.78 & 0.89 & 57.47 & 58.92 & 4.55 & 5.55 & -- & -- & --\\
4 &vertical-bs64 &0.81 & 1.04 & 57.00 & 58.15 & 4.12 & 4.94 & -- & -- &--\\
 
\hline \\
8 &vertical-bs1 &17.20 & 18.70 & 27.30 & 27.84 & 16.77 & 18.05 & 4165217 & 551323.0 & 1939290\\
8 &vertical-bs4 &11.28 & 12.47 & 28.66 & 35.00 & 15.51 & 18.36 & -- & -- & --\\
8 &vertical-bs8 &8.61 & 9.55 & 26.53 & 27.17 & 9.16 & 9.90 & -- &-- & --\\
8 &vertical-bs16 &6.93 & 7.94 & 26.06 & 26.72 & 7.72 & 8.90 & -- & -- & --\\
8 &vertical-bs32 &6.62 & 7.55 & 25.71 & 26.30 & 5.87 & 6.83 & -- & --& --\\
8 &vertical-bs64 &6.19 & 7.02 & 25.30 & 25.97 & 4.84 & 5.70 & -- & -- & --\\
 
\hline \\
16 &vertical-bs1 &92.72 & 96.26 & 20.55 & 21.68 & 18.74 & 21.27 & 17454734 & 1203360.0 & 11294606\\
16 &vertical-bs4 &59.00 & 61.03 & 18.62 & 20.42 & 10.41 & 12.41 & -- & -- & --\\
16 &vertical-bs8 &50.25 & 52.13 & 17.44 & 19.17 & 7.88 & 9.74 & -- & -- & --\\
16 &vertical-bs16 &44.52 & 46.23 & 16.49 & 18.98 & 6.60 & 8.32 & -- & -- & --\\
16 &vertical-bs32 &43.04 & 44.83 & 15.89 & 18.67 & 5.35 & 7.05 & -- & -- & --\\
16 &vertical-bs64 &47.19 & 48.26 & 15.52 & 19.02 & 5.30 & 7.00 & -- & -- & --\\
 
\end{tabular}
\label{tab:20-newsgroups-blocksize}
\end{table*}

\section{Conclusions \& Future Work}
\label{sec:conclusion}

We have designed new parallel algorithms for the efficient practical
algorithms proposed by Bayardo et. al \cite{Bayardo2007Scaling}. We
have compared various optimizations to the practical algorithms, and
we have found that a simple optimization to \var{all-pairs-0} which we call
\var{all-pairs-0-array} gave the best results. We have been able to
distribute both vectors and dimensions in a way that is faithful to
the original processing order and data structures of
\var{all-pairs-0-array}. The vertical parallel algorithm distributes
dimensions and parallelizes the inner loop, accumulating
candidates. We have proposed an effective pruning step to decrease the
number of candidates communicated in this step
(\prettyref{lem:localpruning}). Various optimizations and
implementation choices for the vertical algorithm have been
considered, including a recursive similarity match search
algorithm. The horizontal parallel algorithm is easier and it
parallelizes the outer loop of the algorithm. We have also proposed a
2-D parallel algorithm which combines the inner-loop and outer-loop
parallelizations in an elegant fashion. Our experiments show that the
variety of parallelizations is useful for large-scale similarity graph
construction.

In the future, we would like to incorporate more
techniques to prune candidates, and other optimizations into our
framework. For instance, it may be possible to exploit the Zipf-like
distribution of dimension frequencies, in a better way. 
Data decomposition approaches, like that of \cite{kulkarni06}, may be
incorporated. It may also be worthwhile to investigate the
applicability of our data distribution approach to approximate
similarity search and knn algorithms, as well as different algorithmic
approaches to proximity search. The scalability of both the vertical
and the 2-D algorithms could be improved upon. For the vertical
algorithm, a better recursive local pruning algorithm could be useful,
or more intelligent pruning heuristics could be discovered. For the
2-D algorithm, a better implementation could make use of asynchronous
communication and burst-mode transfers. In general, it is an open
problem to find the best data decomposition for parallel solutions of
this problem which does not suffer from the replication bottleneck of
the horizontal distribution. Our present results may lead to better
solutions in that area, eventually.

\section*{Acknowledgments}

Ata T\"urk provided the real world datasets for the parallel all 
pairs algorithm. Ata T\"urk also reviewed the paper in detail, 
which improved the
organization of the paper, and he fixed a condition in the basic
horizontal algorithm, and thus substantially supported our research. He has also
contributed theoretical ideas that did not make it to this paper, but
are likely to be published in consequent work.
The experiments were performed at the TUBITAK ULAKBIM high
performance computing center. The copyright of the OCaml implementation
belongs to G\"ok Us Sibernetik Ar-Ge Ltd., and it is used in big data
applications of the said company.

\bibliography{all-pairs} 
\bibliographystyle{plain}

\end{document}